\documentclass[sigconf]{acmart}

\AtBeginDocument{%
  \providecommand\BibTeX{{%
    \normalfont B\kern-0.5em{\scshape i\kern-0.25em b}\kern-0.8em\TeX}}}

\copyrightyear{2021}
\acmYear{2021}
\setcopyright{acmcopyright}\acmConference[HSCC'21]{24th ACM International
Conference on Hybrid Systems: Computation and Control}{May 19--21,
2021}{Nashville, TN, USA}
\acmBooktitle{24th ACM International Conference on Hybrid Systems: Computation
and Control (HSCC'21), May 19--21, 2021, Nashville, TN, USA}
\acmPrice{15.00}
\acmDOI{10.1145/3447928.3456657}
\acmISBN{978-1-4503-8339-4/21/05}

\usepackage{amsmath,amsfonts}
\usepackage{graphicx}
\usepackage{textcomp}
\usepackage{xcolor}
\usepackage{bm,tikz}
\usepackage{algorithm}
\usepackage{algpseudocode}
\usepackage{algorithmicx}
\usepackage{appendix}
\usepackage{times}
\usepackage{subfigure} 
\usepackage{braket, multirow}
\usepackage{todonotes}
\usepackage[outdir=./]{epstopdf}
    \newtheorem{theorem}{Theorem}
 \newtheorem{proposition}{Proposition}
 \newtheorem{definition}{Definition}
\newtheorem{example}{Example} 
\newtheorem{problem}{Problem} 
\newtheorem{remark}{Remark}
	\definecolor{caribbeangreen}{rgb}{0.486 0.73 0}
	\definecolor{cblue}{rgb}{0,0.63,0.945}
	\definecolor{cred}{rgb}{0.964 0.325 0.078}
\setcopyright{none}
\renewcommand\footnotetextcopyrightpermission[1]{}

\newcommand{\oomit}[1]{}

\begin{document}

\title{Switching Controller Synthesis for Delay Hybrid Systems\\ under Perturbations}


\author{Yunjun Bai}
\email{baiyj@ios.ac.cn}
\orcid{1234-5678-9012}
\affiliation{
  \institution{SKLCS, Institute of Software, CAS \\
   \& Univ. of CAS, Beijing, China}
}

\author{Ting Gan}
\email{ganting@whu.edu.cn}
\affiliation{%
  \institution{Wuhan  University}
  \city{Wuhan, China}
}

\author{Li Jiao}
\email{ljiao@ios.ac.cn}
\affiliation{%
  \institution{SKLCS, Institute of Software, CAS \\
   \& Univ. of CAS, Beijing, China}
}

\author{Bican Xia}
\email{xbc@math.pku.edu.cn}
\affiliation{
  \institution{Peking University}
  \city{Beijing,China}
}

\author{Bai Xue}
\email{xuebai@ios.ac.cn}
\affiliation{
 \institution{SKLCS, Institute of Software, CAS \\
  \&Univ. of CAS, Beijing, China}
}

\author{Naijun Zhan} 
\email{znj@ios.ac.cn}
\affiliation{
  \institution{SKLCS, Institute of Software, CAS \\
   \&Univ. of CAS, Beijing, China}
 }



\begin{abstract}
    Delays are ubiquitous in modern hybrid systems, which exhibit both continuous and discrete dynamical behaviors. Induced by signal transmission, conversion, the nature of plants, and so on, delays may appear either in the continuous evolution of a hybrid system such that the evolution depends not only on the present state but also on its execution history, or in the discrete switching between its different control modes. In this paper we come up with a new model of hybrid systems, called \emph{delay hybrid automata}, to capture the dynamics of systems with the aforementioned two kinds of delays. Furthermore, based upon this model we study the robust switching controller synthesis problem such that the controlled delay system is able to satisfy the specified safety properties regardless of perturbations. To the end,  a novel method is proposed to synthesize switching controllers based on the computation of differential invariants for continuous evolution and backward reachable sets of discrete jumps with delays. Finally, we implement a prototypical tool of our approach and demonstrate it on some case studies.
\end{abstract}



\maketitle
\pagestyle{plain}

\section{Introduction}
\oomit{
\begin{quote}
 ``You suddenly notice a ball flying toward your head. Your first reaction happens after
a delay. To avoid the ball, you must consider where your head will be after its delayed
response in relation to where the ball will be. '' 
\end{quote}
\rightline{{\rm --- \emph{Steven A. Frank, Chapter 13 of Control Theory Tutorials.}}}
\vspace{1.5mm}}

With the broad applications of cyber-physical systems (CPS) in our daily life, 
the correct design of reliable CPS is getting increasingly important, especially in safety-critical domains such as automotive, med\-icine, etc. Due to the bidirectional 
conversion between analog and digital signals, the periodicity of collecting data by sensors, and executing the commands by actuators, and the data transmission through networks with different bandwidths, etc., 
time delay is becoming ubiquitous and inevitable in CPS, giving rise to the difficulty of CPS design, 
as delays may invalidate the certificates of stability and safety obtained with abstracting them away, even well annihilate control performance.

Generally, two kinds of delays appear commonly in CPS. One is in continuous evolution of systems, resulting in that the evolution not only depends on the current state, but also on the historical states. 
\oomit{For example, the dynamics of the density of certain blood cells are affected by the process of cellular production in the bone marrow, which has been  modelled by the famous Mackey-Glass equation \cite{glass1988clocks}, a nonlinear delay differential equation (DDE).} 
As an appropriate generalization of ordinary differential equations (ODEs), 
delay differential equations (DDEs) are widely used to capture time-delay continuous  dynamical systems. The other one occurs at discrete jumps between different control modes of the underlying systems. 

\oomit{Recently, more and more research has been studied on the verification and design of delay systems.
In \cite{pola2010symbolic} and \cite{pola2015symbolic},  symbolic abstraction-based methods are introduced for time-invariant and time-varying delay systems by approximate functional space using spline analysis,
however, (in)stability of the concrete systems may be hidden.
Considering the stability of systems, the barrier certificate method has been used to verify delay systems in \cite{prajna2005methods}, with the limitation of given 
specified polynomial templates. 
In \cite{huang2017}, the authors propose a bounded verification method for nonlinear networks with discrete delays. 
However, the dynamics of each subsystem  modelled by ODEs and the analysis is done over a finite time horizon. 
In addition, the well-developed and commonly used methods in industry, numerical simulation, fails to guarantee the satisfaction of desired properties such as safety and reachability. 
Due to the fact that the systems with delays have states of infinite dimension, the methods developed for ODEs or hybrid automata without delays are not able to handle DDEs or delay hybrid systems directly.
 In all these existing works, only one kind of delay is considered, either continuous or discrete delay.
Furthermore, there is still a lack of  appropriate formal models to handle both situations in one system.  }

\oomit{Appropriate switching guard conditions  play an important role in correctness guarantees regarding the dynamics of undergoing safety-critical hybrid systems.}

In this paper, we propose a new model of hybrid systems, called \emph{delay hybrid automata} (dHA), which is an extension of classical hybrid automata (HA)  \cite{henzinger1995s}, in order to capture the dynamics of systems involving the aforementioned two kinds of delays. Based on the proposed dHA, we investigate the safe switching controller synthesis problem for delay hybrid systems, i.e., given a dHA $\mathcal{H}$ and a safety property $\mathcal{S}$, to synthesize a refined dHA 
  $\mathcal{H}^*$ by strengthening the invariant in each mode and the guard condition for each discrete jump 
  such that $\mathcal{H}^*$ satisfies $\mathcal{S}$ robustly, with additional condition that $\mathcal{H}^*$ is \emph{non-blocking} if $\mathcal{H}$ is non-blocking. Our approach is invariant-based, 
 which is  
 a classical approach to synthesizing safe switching controllers   
   for HA  \cite{871306, zhao2013synthesizing}. However,  the computation of differential invariants (the definition will be given in Section~\ref{sec:invariant})
    for the DDE in each mode as well as a global invariant (the definition will be given in Section~\ref{sec:framework}) among these modes is 
    much involved than the counterparts in HA when the two kinds of delays are considered. 
   To compute differential invariants for DDEs, we propose 
   a two-step approach: the first step is to reduce 
   differential invariant generation problem to $T$-differential invariant generation problem 
   using global ball-convergence condition derived in terms of Metzler matrix for a class of linear DDEs, where $T$ is a bounded time horizon; the second step is to obtain an over-approximation of  the $T$-bounded reachable set  
    based on the \emph{growth bound} adapted from \cite{Reissig2016}. Non-linear DDEs can be reduced to the linear case 
    by means of the linearization technique, in case that global ball-convergence is replaced 
    by local ball-convergence. 
    A global  invariant is generated based on fixed  point iteration, and 
    the computation of differential invariants for continuous evolution in each mode 
    and backward reachable sets for discrete jumps by taking delays into account, 
    which  is similar to compute reachable sets of HA, e.g., with  
     dReach \cite{Kong15}. Our approach is finally illustrated on some interesting case studies. 

The main contributions of this work are summarized below:
\begin{itemize}
    \item[(1)]  a new model language, called dHA, is proposed to model delay hybrid systems, which exhibit delays in both continuous- and discrete-time dynamics. 
    \item[(2)] in this new model dHA, a novel approach based on the computation of differential invariants is proposed to address the switching controller synthesis problem for delay hybrid  systems, such that the controlled delay hybrid system is able to satisfy the specified safety property. 
\end{itemize}

\subsection{Related Work}
 Controller synthesis through correct-by-construction manner 
 provides mathematical guarantees to the correctness and reliablity of  (hybrid) 
systems. 
In the literature, this problem has been extensively studied and 
various approaches have been proposed, which can be 
categorized into abstraction based, e.g., 
\cite{TabuadaBook,BeltaBook,Girard12,Reissig2016,DBLP:journals/deds/NilssonOL17,HSCC2018}, 
and constraint solving based, e.g., \cite{zhao2013synthesizing,TalyT10}. 
The basic idea of abstraction based approaches is to  abstract the original system under consideration to a finite-state two-players  game, and  then solve   reactive  synthesis  using automata-theoretic  algorithms  with  respect  to temporal  control  objectives. In contrast, 
the basic idea of  constrains solving based approaches is to  reduce the synthesis problem to 
an invariant generation problem, which can be further reduced to a constraint solving problem. 
As a generalization of \cite{TalyT10},  an optimal switching controller synthesis 
is investigated  in \cite{6064517} by solving an unconstrained numerical optimization problem.
Based on reachable set computation and fixed point iteration, a general framework of controller synthesis for HA is proposed in \cite{871306, 871303}.
However, all these existing 
works focus on ODEs, therefore cannot be applied to 
 DDEs, let alone delay hybrid systems directly. This is because   
  ODEs  are Markovian,  
 but DDEs are non-Markovian, whose states are functionals  with infinite dimension.
In \cite{Chen18,Chen20}, a controller synthesis problem for time-delay discrete dynamical systems  was first investigated by reduction to solving  
imperfect two-player safety game, but 
it is unclear whether their approach can be extended to time-delay continuous 
dynamical systems and delay hybrid systems. 

Recently, verification and synthesis for time-delay systems attract 
increasing attention, we just name a few below.  
Prajna and Jadbabaie extended the notion of \emph{barrier certificate} to 
time-delay systems  \cite{PJ05}. In \cite{Zou15}, Zou et al. 
first proposed interval Taylor model for DDEs, and then  discussed automatic stability analysis and 
safety verification  based on interval Taylor model  and stability analysis of discrete dynamical systems. However, their approach can only be applied to specific DDEs, 
  whose right sides are independent of current states. Following this line, more efficient algorithms for analyzing Taylor models to inner and outer approximate reachable sets of more general DDEs in finite time horizon were 
given  \cite{goubault2019inner}. In \cite{feng2019taming}, 
Feng et al. further considered how to utilize stability analysis of linear delay dynamical 
systems and linearization  to reduce the unbounded verification  
to the bounded verification for a class of general DDEs. 
Based on  \cite{feng2019taming}, \cite{Bai2021} investigated switching controller synthesis problem 
of delay hybrid systems, in which time-delay in discrete jumps is not taken into account.
In contrast, the approach proposed in this paper 
can compute differential invariants for DDEs using ball-convergence based on Metzler matrix analysis,  growth bound and linearization, it could be more powerful and applied to 
verify more DDEs (see Example~\ref{ex:heating}). 
In \cite{Chen16}, a simulation-based approach to approximate reachable sets of ODEs was 
extended to DDEs.  Meanwhile,  a topological homeomorphism-based approach was proposed to over- and under-approximate reachable sets of a class of DDEs  \cite{xue2017safe}. 
Later, this approach was further extended to deal with perturbed DDEs in \cite{xue2020}.
Like \cite{goubault2019inner}, these approaches can only be applied to compute 
reachable sets in finite time horizon. 
 In addition, in \cite{pola2010symbolic,pola2015symbolic}, Pola et al. proposed approaches how to construct symbolic abstractions  for time-invariant and time-varying delay systems by approximating  functional space using spline analysis.
In \cite{huang2017}, Huang et al.  proposed a bounded verification method for nonlinear networks with discrete delays. 
Nonetheless, the dynamics of each subsystem  modelled by ODEs and the analysis is done over a finite time horizon. 
Evidently, only one kind of delays is considered in all these existing works,  either continuous or discrete.
There is indeed a lack of  appropriate formal models to handle both situations uniformly. 

\subsection{A Motivating Example}
\begin{figure}[t]
    \centering
    \includegraphics[width=4.5cm, height=4cm]{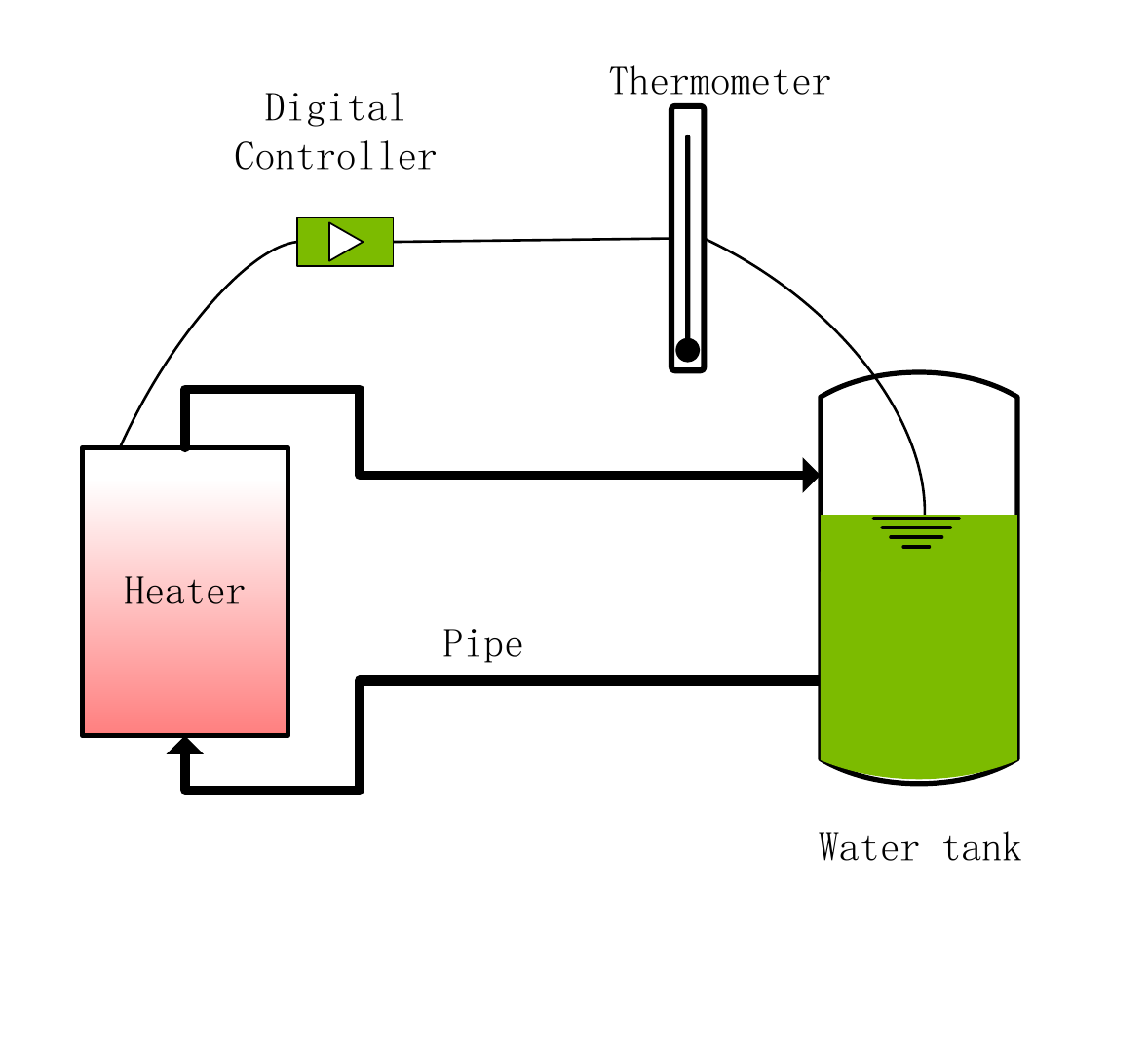}
    \vspace*{-7mm} 
    \caption{A heating system}
     \label{fig:heating}
\end{figure}
To illustrate the main idea of our approach, we use a  heating system
as a motivating example, as depicted in Fig.~\ref{fig:heating}, consisting of  the following four components:
\begin{itemize}
    \item[(1)] a water tank with water,
    \item[(2)] a heater with on and off two states,
    \item[(3)] a thermometer monitoring the temperature of the water in the tank, and 
echoing  warning signals whenever the temperature of the water is above or below certain thresholds,
    \item[(4)] pipes connecting the heater and the tank.
\end{itemize}
Additionally, we add a controller that observes the signals 
produced in the thermometer, and computes a command to the heater in order to
maintain the temperature of the water within a given range.
The temperature of water in the tank is desired to stay between $20$ and $90$ degrees through switching the heating on and  off.
The behavior of the  temperature of water  in the tank is mixed continuous evolution 
with discrete switches, which can be modelled  by a hybrid automaton \cite{ALUR19953}.
However, the delay impact of pipes and thermometer monitoring are both neglected in
these models. 
In \cite{kicsinya2012heating}, it was pointed out that 
energy efficiency can be  increased by $5-10\%$ if the delay impact of pipes
is considered. 
Moreover, due to the delay possibly caused by 
measuring the thermometer, sending the signals, 
executing the control commands and so on, the temperature of 
water in the tank could be beyond the thresholds, which is definitely unsafe. 
Therefore, the delay impacts  of the pipes and the thermometer have to be
 taken into account  when 
we model  the temperature of water in the tank.

\subsection{Basic Notations and Definitions}\label{sec:pre}
\textbf{Notations.}
Let $\mathbb{N}$, $\mathbb{R}$ and $\mathbb{C}$ be the set of natural, real and complex numbers, 
$\mathbb{R}_+$ be  the set of positive real numbers.
For $z=a+ib\in \mathbb{C}$ with $a,b\in \mathbb{R}$, $\Re(z)=a$ and $\Im(z)=b$, respectively, denote the real and imaginary parts of $z$.
$\mathbb{R}^n$ is the set of $n$-dimensional real vectors, denoted by boldface letters.
Given a vector $\mathbf{x}\in \mathbb{R}^n$, $x_i$ denotes the $i$-th coordinate of $\mathbf{x}$
for $i\in \{1,2,\ldots,n \}$, and its maximal norm is $\|\mathbf{x}\|_{\infty}=\max_{1\leq i \leq n} |x_i|$.
For a vector $\mathbf{y} \in \mathbb{R}^n_+$, let $(\mathbf{y})_{\min}=\min_{1\leq i \leq n} y_i$.
Given two vectors $\mathbf{x} , \mathbf{y} \in \mathbb{R}^n$, we define $\mathbf{x}\geq \mathbf{y}$ iff  $x_i \geq y_i$ for all $1\leq i \leq n$,
and   $\mathbf{x} > \mathbf{y}$ iff  $x_i > y_i$ for all $1\leq i \leq n$.
Given $\epsilon >0$, we define $\mathfrak{B}(\epsilon) = \{ \mathbf{x}\in \mathbb{R}^n
\mid \|\mathbf{x}\|_{\infty} \leq \epsilon \}$ as the $\epsilon$-closed ball around $\mathbf{0}$.
Let $\mathbb{R}^{n\times m}$ be the set of real $n \times m$ matrices. 
The entry in the $i$-th row and $j$-th column of a matrix $M\in \mathbb{R}^{n\times m}$ is denoted as $m_{ij}$ with $1\leq i \leq n$ and $1\leq j\leq m$.
\oomit{For any matrix $M \in \mathbb{R}^{n\times n}$, the spectrum of $M$ is denoted by 
$\delta(M) = \{\lambda \in \mathbb{C} \mid \det(\lambda I -M)=0 \}$,
where $I$ is the $n\times n $ identity matrix.
And the spectral abscissa  of $M$ is defined by $\mu(M) =  \max\{ \Re(\lambda) \mid\lambda \in \delta(M)\}$.}
For $t_1 < t_2$, $\mathcal{C}\{ [t_1, t_2], \mathbb{R}^n\}$ is the space of continuous functions from $[t_1, t_2]$ to $\mathbb{R}^n$.
For a set $\mathbb{A} \subseteq \mathbb{R}_+^{n}$, $\mathbf{a} = \sup \mathbb{A}$ iff for all $\mathbf{x} \in \mathbb{A}$, $\mathbf{x} \leq \mathbf{a}$, and for any  upper bound $\mathbf{y} \in \mathbb{R}_+^n$, then $\mathbf{y}\geq \mathbf{a} $.
Finally, we denote $(x)^+ = \max(0,x)$ for any real number $x\in \mathbb{R}$. 

In this paper, we consider a class of time-delay systems under perturbations described as follows:
    \begin{equation}\label{dde:DDE} 
      \hspace*{-3mm}   \left\{  
                     \begin{array}{lr}  
                     \dot{\mathbf{x}}(t)=\bm{f}(\mathbf{x}(t), \mathbf{x}(t-r_1), \ldots, \mathbf{x}(t-r_k), \mathbf{w}(t)), \hspace*{-2mm}  & \hspace*{-2mm}  t\in [0,\infty) \\  
                     \mathbf{x}(t)= \bm{\phi}(t) , & t\in [-r_k,0]
                     \end{array}  
        \right.  
    \end{equation}
where $\mathbf{x} \in \mathbb{R}^n$ is the state vector, $t\in \mathbb{R}$ models time, 
The discrete delays are assumed to satisfy $0 < r_1 < r_2 < \cdots < r_k$.
$\mathbf{w}(\cdot): [0,\infty)\mapsto \mathbb{R}^m$ is external disturbance vector, which is unknown but assumed to be bounded by 
a given constant $w_{max}$, i.e., $\| \mathbf{w}(t) \|_{\infty} \leq w_{max}$ for all  $t\geq 0$.
 $\bm{\phi}(\cdot)\in \mathcal{C}\{ [-r_k, 0], \mathbb{R}^n\}$ is the initial condition. 
Suppose that $\bm{f}$ is continuous and satisfies the Lipschitz condition, 
then from a given initial condition $\bm{\phi}$ and $\mathbf{w}(t)$, there exists a unique solution $\bm{\xi_{\phi}}^{\mathbf{w}}(\cdot): [-r_k, \infty)\mapsto \mathbb{R}^n$.

\begin{definition}[Metzler matrix\cite{berman1994nonnegative}]\label{def:metzler}
A matrix $M \in \mathbb{R}^{n\times n}$ is called a Metzler matrix if all off-diagonal elements of $M$ are non-negative, 
i.e.,
$m_{ij} \geq 0$ whenever $i\neq j$. 
\end{definition}

Regarding Metzler matrices, the following proposition holds, please refer to \cite{berman1994nonnegative} for the detail.
\begin{proposition}[\cite{berman1994nonnegative}] \label{pro:Metzler}
For any Metzler matrix $M$, 
the following two properties are equivalent 
\begin{itemize}
\item[1.] $\mu(M) <0$, where $\mu(M)=\max\{ \Re(\alpha) \mid \alpha \in \mathbb{C}: \det(\alpha \mathcal{I}-M)=0)\}$, $\mathcal{I}$ is the $n\times n$ identity matrix.
\item[2.] there exists $\bm{\zeta} \in \mathbb{R}^n$ and $\bm{\zeta} > \bm{0}$ such that $ M \bm{\zeta}< \bm{0}$.
\end{itemize}
\end{proposition}

The structure of this paper is organized as: 
the notion of delay hybrid automata and the safe switching controller synthesis problem
of interest are defined in Section \ref{sec:da}.
After presenting an approach for invariant generation of delay hybrid systems in Section \ref{sec:invariant}, Section \ref{sec:synthesis} concentrates on the controller synthesis framework based on the global invariants generation  for  delay hybrid systems. We demonstrate our approach with  two examples in Section \ref{sec:expri}. Finally Section \ref{sec:conclu} concludes this paper.

\section{Delay Hybrid Automata and Problem Statement}\label{sec:da}
Hybrid automata (HA) \cite{henzinger1995s} are popular models for dynamical systems with complex mixed continuous-discrete behaviors. In order to characterize  behaviors of 
 hybrid systems with 
the two type of time delays aforementioned, we introduce an extension of HA,  called \emph{delay hybrid automata} (dHA),  
formally defined  as follows: 
\oomit{ continuous evolution of a delay hybrid system apparently can be modeled by DDE \eqref{dde:DDE}. 
 Meanwhile,
 discrete jumps between different control modes of the system may take certain time,  which can be simply described  by weights associated with each edge in a hybrid automaton. 
Along with both adjustments, other potential impact resulting from two kinds of delays is studied in more details and summarized  in Definition \ref{def:ad}.}
\begin{definition}[Delay hybrid automaton, \mbox{dHA}]
    A dHA is a tuple $\mathcal{H}=(Q,X, U, I, \Xi,F,E, D, G, R)$, where,
    \begin{itemize}
        \item $Q=\{ q_1,\ldots,q_m \}$ is a finite set of modes;
        \item $X$ is a set of state variables; 
        \item $U \subseteq \mathcal{C}\{[t_1, t_2], \mathbb{R}^n\}$, where $t_1 < t_2$, is a set of continuous functionals;
        \item  $I: Q\mapsto 2^{ \mathbb{R}^n} $ gives each mode $q\in Q$ an invariant $I(q)\subseteq \mathbb{R}^n$;
        \item $\Xi: Q \mapsto 2^U$ gives each mode $q\in Q$ its initial states set  $\Xi(q)\subseteq  U$;
        \item $F=\{\bm{f}_{q_1}, \ldots, \bm{f}_{q_m}  \}$ is the set of vector fields, each mode $q\in Q$ has unique vector field $\bm{f}_q$, which  is used to 
        form a delayed differential equation \eqref{dde:DDE} to model the continuous evolution, i.e., $$\dot{\mathbf{x}}(t)=\bm{f}_q(\mathbf{x}(t), \mathbf{x}(t-r^q_{1}), \ldots ,\mathbf{x}(t-r^q_{k}),\mathbf{w}(t) );$$ 
        \item $E\subseteq Q \times Q$ is the set of discrete transition relations between modes;
        \item  $D: E\mapsto \mathbb{R}_+ $ gives each  discrete transition $e\in E$ a delay time $D(e)\in \mathbb{R}_+$;
        \item $G: E\mapsto 2^{\mathbb{R}^n}$ denotes guard conditions; 
        \item $R: E \times X_D \mapsto U$ denotes reset functions. 
        \end{itemize}
        \label{def:ad}
\end{definition}
 Compared with the definition of HA, 
there are several notable changes  in Definition \ref{def:ad}:
a new item $U\subseteq \mathcal{C}\{ [-r_k^q, 0], \mathbb{R}^n\}$ is introduced to represent the set of all possible initial states. Note that the solution to a DDE is a functional, and correspondingly a state is a function standing 
the execution history up to the considered instant starting from the given initial state, rather than a point in $\mathbb{R}^n$ as for ODE.  
\oomit{continuous functionals, i.e., $U=\{ \mathbf{x}_{t}^{\bm{\phi}}(\cdot) \in 
	 \mathcal{C}\{ [-r_k^q, 0], \mathbb{R}^n \} \}$, where  $\mathbf{x}_{t}^{\bm{\phi}}(\theta)=\bm{\xi}_{\bm{\phi}}^{\mathbf{w}}(t+\theta)$ for $\theta \in [-r_k^q ,0]$ with respect to DDE \eqref{dde:DDE}.
That is because the vector field of each control mode
is a DDE and its initial condition 
is a continuous function  (corresponding to $\Xi$ and $F$).}
Additionally, another new item $D$ is used to specify the delays in discrete transitions: for each $e=(q, q')\in E$, the delay is denoted by $D(e)\in \mathbb{R}_+$.
Moreover,  the reset function $R$ is changed to  $E \times X_D \mapsto U$ accordingly, where $X_D$ is the set of  reachable states  
satisfying the corresponding guard condition.
Intuitively, when a mode switching happens, e.g.,  a transition from $q$ to $q'$  at time $t$, there exists time $\theta\in [-r_k^q, 0]$, the system has to satisfy: 
$\mathbf{x}_{t}^{\bm{\phi}}(\theta) \in G(e)$, and the update state is $ \bm{\phi}' = R(e,\mathbf{x}_{t+D(e)}^{\bm{\phi}}(\cdot) )$. 

\oomit{So, 
 a reset map in classical hybrid automata can be seen as a special case of delay hybrid automata, just setting R(f1(t))=f2(t).
in dHA,  a reset function maps a functional state to a functional state, can be seen as a conservative extension of the concept of reset functions in classical HA}

\oomit{ For a delay hybrid system, in each mode, its vector field is a DDE. Mathematically, a time-dependent solution of a DDE cannot be uniquely determined by its initial state at a given moment alone unlike ODEs, instead, the solution profile (initial function) on an interval of length equal to the maximal delay prior to the starting time t (normally taking t = 0) has to be prescribed. That is, we need to define continuous initial conditions over [–r,t] and hence DDEs are infinite-dimensional systems. So, any reset function should meet this respect conceptually. Actually, we can understand a reset map R to map f1(t) to a function f2 over [-r_q,t], where f1(t) stands for the solution of the DDE in the pre-location at the jump instant t, and f2 is the initial condition over [-r_q,t] of the DDE in the post-location with starting instant t, and r_q is the maximal delay in the DDE of the post-location. So, a reset map in classical hybrid automata can be seen as a special case of delay hybrid automata, just setting R(f1(t))=f2(t).}

\begin{example}\label{exm:heating_automaton}
    For the heating system  shown in the motivating example, it is straightforward to present  its dHA textually as follows:  
\begin{itemize}
    \item $Q=\{ q_1,q_2 \}$;  (two modes of discrete states, heater on and off);
    \item $X =\{x\}$; (the temperature of water in the tank);
    \item $U =\mathcal{C}$; (all continuous functionals);
    \item $I(q_1)=\{x\in \mathbb{R}\mid 20\leq x\leq 90\}$ and $I(q_2)=\{x\in \mathbb{R}\mid 20\leq x\leq 90\}$;
    \item $\Xi(q_1) = \{ x(t) \mid x(t)=50-10\sin{t}, t\in [-1, 0]\}$ and $\Xi(q_2) = \{ x(t) \mid x(t)=85-5\sin{t}, t\in [-1,0]\}$; 
    \item $F=\{f_{q_1}, f_{q_2}\}$, where $f_{q_1} =K_1(h - x(t)) +K_2 x(t-1) +w_1$ and $f_{q_2} = -K_1 x(t) + K_2 x(t-1) + w_2$, $K_1$, $K_2$, $h$, $w_1$ and $w_2$ are real constants. That is, 
           the temperature rises and decreases following the respective  DDE in 
            $q_1$ and $q_2$, respectively; 
    \item $E=\{e_1=(q_1, q_2), e_2=(q_2, q_1)\}$; 
    \item $D(e_1)= 2$ and $D(e_2)=2$;
    \item $G(e_1)=\mathbb{R}$ and $G(e_2)=\mathbb{R}$; 
    \item $R(e_1,x_{t+D(e_1)}(\cdot))=x(\theta),\theta\in [t+D(e_1)-1, t+D(e_1)] $ with $x(t)\in G(e_1)$ and $R(e_2,x_{t+D(e_1)}(\cdot))= x(\theta),\theta\in [t+D(e_2)-1, t+D(e_2)]$ with $x(t)\in G(e_2)$.  
\end{itemize}
Pictorially, the dHA is shown in Fig.~\ref{fig:heating_automaton}.
\begin{figure}[t]
  \centering
  \usetikzlibrary{positioning,automata}
\begin{tikzpicture}[shorten >=1pt,node distance=4cm,on grid,line width=0.3mm]
  \tikzstyle{every state}=[fill=caribbeangreen!70]
  \node[state,initial]   (q_0)               
  {\begin{minipage}{0.04\textwidth}  
                    \begin{align*}       
                      \dot{x}= f_{q_1} \\
                      x\leq 90
                    \end{align*}
                \end{minipage}};
  
  \node[state] (q_2) [right=of q_0] 
    {\begin{minipage}{0.04\textwidth}  
    \begin{align*}       
         \dot{x}= f_{q_2}  \\
         x\geq 20
    \end{align*}
    \end{minipage}};
  \path[->] (q_0) edge [bend left=30] node [above] {off} node[below]{ $D(e_1) = 2$} (q_2)
            (q_2) edge [bend left=30]  node [above] {on} node[below]{ $D(e_2) = 2$} (q_0);    
\end{tikzpicture}
\caption{The dHA for the heating system}
\label{fig:heating_automaton}
\end{figure}
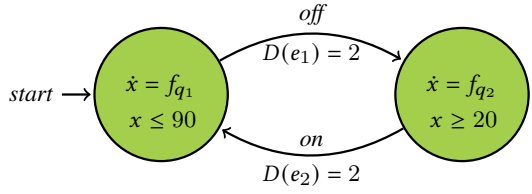
\end{example}


\begin{definition}[Hybrid execution]
   For a dHA $\mathcal{H}$, given an initial hybrid state $(q_0, \bm{\xi}^{\mathbf{w}}_{\bm{\phi}_{0}}(0))$ and $\mathbf{w}(\cdot):[0,\infty)\mapsto \mathbb{R}^m$, an execution $\pi$ of the delay hybrid automaton $\mathcal{H}$
is a sequence of $\langle t_i, q_i, \bm{\xi}_{\bm{\phi}_{i}}^{\mathbf{w}}(t_i) \rangle$, for $i\in \mathbb{N}$ and $q_i\in Q$,
satisfying  that 
    any transition $ \langle t_i, q_i, \bm{\xi}^{\mathbf{w}}_{\bm{\phi}_{i}} (t_i)\rangle \mapsto  \langle t_{i+1}, q_{i+1}, \bm{\xi}_{\bm{\phi}_{i+1}}^{\mathbf{w}}(t_{i+1}) \rangle $ is either :
     \begin{itemize}
        \item \emph{the continuous evolution}: 
         $q_i = q_{i+1}$, $\bm{\phi}_{i}=\bm{\phi}_{i+1}$, $t_i < t_{i+1}$, and for all $t\in [t_i, t_{i+1}]$, the solution of DDE $\dot{\mathbf{x}}=\bm{f}_{q_i}$ is $\bm{\xi}_{\bm{\phi}_{i}}^{\mathbf{w}}(\cdot):
        [t_i, t_{i+1}] \mapsto \mathbb{R}^n$, and  $\bm{\xi}_{\bm{\phi}_{i}}^{\mathbf{w}}(t) \in I(q_i)$;
      \item \emph{the discrete transition}: 
        $e=(q_i,q_{i+1})\in E$, $t_i=t_{i+1}$, and there exists $t$ such that $t_{i+1}=t+D(e)$ and $\bm{\xi}_{\bm{\phi}_{i}}^{\mathbf{w}}(t) \in G(e)$ and $\bm{\phi}_{i+1} = R(e,\mathbf{x}_{t+D(e)}^{\bm{\phi}_{i}}(\cdot))$.
     \end{itemize}
    
 
\end{definition}

An execution $\pi$ is called \emph{finite} if it is a finite sequence ending with a closed time interval.
Otherwise, the execution $\pi$ is called \emph{infinite} if it is an infinite sequence or if $\sum_{i=0}^{N}(
t_{i+1}-t_{i})=\infty$, where $N\in \mathbb{N}$.
A dHA $\mathcal{H}$ is called \emph{non-blocking} if there exists at least one infinite execution starting from any initial state.

\begin{definition}[Reachable set]
    Given a dHA $\mathcal{H}$, the reachable set $\mathrm{R}_{\mathcal{H}}(t)$ for the delay hybrid system  
    within  $[-r^{q_0}_{k}, t]$ is 
    \begin{equation*}
     \mathrm{R}_{\mathcal{H}}(t) =
     \left\{
          \mathbf{x}(t)  \,\middle |\,
          \begin{array}{cc}
          \forall \, q_0\in Q,\, \forall \, \bm{\xi}^{\mathbf{w}}_{\bm{\phi}_{0}}(0)\in \Xi(q_0), \\
          \exists \,  \pi =\langle t_0, q_0, \bm{\xi}^{\mathbf{w}}_{\bm{\phi}_{0}}(0)\rangle, \cdots, \langle t,q_i,\bm{\xi}_{\bm{\phi}_i}^{\mathbf{w}}(t)\rangle \\
        s.t.\,  \mathbf{x}(t)= \bm{\xi}_{\bm{\phi}_i}^{\mathbf{w}}(t)
         \end{array}
     \right\}.
    \end{equation*}
\end{definition}

\begin{example}
    An execution  for the heating system in the motivating example is given below.
\begin{figure}[h]
    \centering
     \begin{tikzpicture}[node distance=2cm]
        \node[rectangle,rounded corners=0.15cm,fill=gray!30, draw=gray!50] (A) at (0, 0) {$\langle 0,q_1, x=50.00 \rangle$};
        \node[rectangle,rounded corners=0.15cm,fill=caribbeangreen!70, draw=caribbeangreen] (B) at (0, -1) {$\langle 10, q_1, x=68.86\rangle$};
        \node[rectangle,rounded corners=0.15cm,fill=cred!70, draw=cred] (C) at (0, -2) {$\langle 12,q_1, x=71.12\rangle$};
        \node[rectangle,rounded corners=0.15cm,fill=gray!30, draw=gray!50] (D) at (3, -2) {$\langle 12,q_2, x=71.12\rangle$};
        \node[rectangle,rounded corners=0.15cm,fill=caribbeangreen!70, draw=caribbeangreen] (E) at (3, -1) {$\langle 17, q_2, x=45.37\rangle$};
        \node[rectangle,rounded corners=0.15cm,fill=cred!70, draw=cred] (F) at (3, 0) {$\langle 19, q_2, x=43.59\rangle$}; 
        \node[rectangle,rounded corners=0.15cm,fill=gray!30, draw=gray!50] (G) at (6, 0) {$\langle 19, q_1, x=43.59\rangle$};
        \node[rectangle,rounded corners=0.15cm,fill=gray!30, draw=gray!50] (H) at (6, -1) {$\langle 23, q_1, x=54.82\rangle $};
        \node[rectangle,rounded corners=0.15cm,fill=gray!30, draw=gray!50] (I) at (6, -2) {$\cdots$};
      \draw[->, line width=0.7pt](A)--node[anchor=center,text=black, left] {\small \textbf{(1)}} (B); 
      \draw[->, line width=0.7pt](B)--node[anchor=center,text=black, left] {\small \textbf{(2)}} node[anchor=center,text=black, right] {\small $D(e_1)=2$}(C) ;
      \draw[->,  line width=0.7pt] (C)--node[anchor=center,cblue, above] {\small \textbf{(3)}}(D) ;
      \draw[->, line width=0.7pt](D)--node[anchor=center,text=black, left] {\small \textbf{(4)}}(E); 
      \draw[->, line width=0.7pt](E)--node[anchor=center,text=black, left] {\small \textbf{(5)}} node[anchor=center,text=black, right] {\small $D(e_2)=2$}(F); 
      \draw[->, line width=0.7pt](F)--node[anchor=center,text=cblue, above] {\small \textbf{(6)}} (G); 
      \draw[->, line width=0.7pt](G)--node[anchor=center,text=black, left] {\small \textbf{(7)}}(H);
      \draw[->, line width=0.7pt](H)--(I);
    \end{tikzpicture}
    \label{fig:excution_heating}
\end{figure}

From the initial state $(q_1, x=50.00)$, the system reaches the state $(q_1, x=68.86)$ in green after $10s$, which is indicated by transition $(1)$. 
Assume that the state $(q_1, x=68.86)$ in green satisfies the guard condition, the system chooses to  jump from mode $q_1$ to mode $q_2$. However, there is a delay $D(e_1) =2$ incurred by the edge $e_1$. The 
system keeps evolving in mode $q_1$ until hitting the state $(q_1, x=71.12)$ revealed by transition $(2)$, and completes the switching by reaching the state $(q_2, x=71.12)$ displayed by transition $(3)$ in blue. Continue this execution as above.
\end{example}

\begin{definition}[Safety]
Given a dHA $\mathcal{H}$ with a safe set $\mathcal{S} = \cup_{q\in Q}  \mathcal{S}_q$, where  $\mathcal{S}_q\subseteq \mathbb{R}^n$, the automaton $\mathcal{H}$ is \emph{$T$-safe} with respect to $\mathcal{S}$ in time $T$, if for any time $t\in [-r^q_{k},T]$, all reachable states $\mathrm{R}_{\mathcal{H}}(t)$ of the system  starting from any initial states  are contained in $\mathcal{S}$, i.e., \[ \mathrm{R}_{\mathcal{H}}(t) \subseteq \mathcal{S},   \forall t\in[-r^q_{k}, T].\]
If $T$ is infinite, then the dHA is \emph{safe} over the infinite-time horizon.
\end{definition}


Now, the problem of interest can be formally formulated as follows: 
\begin{problem}[\textbf{Safe Switching Controller Synthesis Problem}] \label{prob:synthesis}
Given a dHA $\mathcal{H}=(Q,X, U, I, \Xi,F,E,D, G, R)$
and a safety property $\mathcal{S}$, the switching controller problem is to synthesize
a new dHA $\mathcal{H}^*=(Q,X, U^*, I^*, 
\Xi^*,F,E,D, G^*, R)$ such that $\mathcal{H}^*$
satisfies:
    \begin{enumerate}
        \item[(r1)]  $\mathcal{H}^*$ is safe, i.e. in $[-r^q_{k}, \infty)$, the reachable set $\mathrm{R}_{\mathcal{H}^*} \subseteq \mathcal{S}$.
        \item[(r2)] $\mathcal{H}^*$ is a refinement of $\mathcal{H}$, i.e., 
        it holds:  $\Xi^*\subseteq \Xi \cap \mathcal{S}$, $I^*\subseteq I$, $U^*\subseteq U$, 
        and for any $e\in E$, it holds: $\forall \mathbf{x}(t)\in G^*(e)$, $\mathbf{x}(t+D(e)) \in G(e)\cap I^*(q)$.
        \item[(r3)]  if $\mathcal{H}$ is non-blocking in the safe set $\mathcal{S}$, then $\mathcal{H}^*$ is non-blocking. 
    \end{enumerate}
$SC=\{ G^*(e)\subseteq \mathbb{R}^n \mid e\in E \}$ is called a \emph{safe switching controller} of $\mathcal{H}$,
if $\mathcal{H}^*$ satisfies above three requirements. 
We call $SC$ is a \emph{trivial switching controller} of $\mathcal{H}$, if there exists one mode $q\in Q$ or one edge $e\in E$ with $I^*(q) = \emptyset$ or $G^*(e) =\emptyset$. 
\end{problem} 
\section{Differential Invariant Generation}\label{sec:invariant}
\emph{Differential invariant} generation plays a central role 
in our  framework to synthesize switching controllers 
for delay hybrid systems with perturbations. 
In this section, inspired by the work in \cite{feng2019taming}, we 
present  a two-step procedure to synthesize differential invariants 
for a delay dynamical system. 
The first step is to calculate a bounded horizon $T$ using ball convergence analysis, which reduces the differential invariant generation problem to the $T$-differential invariant generation problem. The second step is to compute an over-approximation of the 
reachable set in time $T$, which is a $T$-differential invariant. 

We first develop the aforementioned two-step method for linear delay dynamical systems, and then generalize it to nonlinear delay dynamical systems.

\begin{definition}[Differential invariant]\label{def:diff_invariant}
Given a mode $q\in Q$ of a delay hybrid automaton $\mathcal{H}$: $(\Xi(q), \bm{f}_q, I(q))$ and time $T$, a set $I^*(q)$ is called a $T$- invariant 
if for any trajectory starting from a given initial function $\bm{\phi}(t) \in \Xi(q)$, $t\in [-r^q_{k},0]$, the following condition holds for $\mathbf{w}(\cdot): [-r^q_{k}, T] \mapsto \mathbb{R}^m$:
\begin{equation*}
    \forall t\in [-r^q_{k}, T], \  \bm{\xi}_{\bm{\phi}}^{\mathbf{w}}(t) \in I(q) \implies  \forall t\in [-r^q_{k}, T], \bm{\xi}_{\bm{\phi}}^{\mathbf{w}}(t)\in I^*(q). 
\end{equation*}
If $T$ is infinite, then $I^*(q)$ is a differential invariant of mode $q$.
\end{definition}

$T$- invariant $I^*(q)$ requires that every trajectory starting from initial set $\Xi(q)$
in time $T$ remains inside the differential invariant $I^*(q)$ if it remains in the domain $I(q)$.
A \emph{safe} differential invariant requires $I^*(q)\subseteq \mathcal{S}_q$.


\subsection{Linear Systems}\label{subsec:linear}

We consider linear  DDEs with the form  \eqref{dde:DDE} first, i.e., 
\begin{equation}\label{dde:linearDDE} 
    \left\{  
                 \begin{array}{lr}  
                 \dot{\mathbf{x}}(t)= A\mathbf{x}(t) + B \mathbf{x}(t-r) + C \mathbf{w}(t), & t\in [0,\infty) \\  
                 \mathbf{x}(t)= \bm{\phi}(t) , & t\in [-r,0]
                 \end{array}  
    \right. \, ,  
\end{equation}
 where $A, B\in \mathbb{R}^{n\times n}$ and $C \in \mathbb{R}^{n\times m}$ are real matrices with  appropriate dimensions.

\begin{definition}[Global ball-convergence]\label{def:ball_convergence}
    Given 
     a  $\mathfrak{r} >0$, 
     \eqref{dde:linearDDE} is called \emph{globally exponentially convergent} within the ball
    $\mathfrak{B}( \mathfrak{r})$, if there exist a constant $\gamma > 0$ and a non-decreasing function $\kappa(\cdot)$ 
    such that 
    \begin{equation*}
        \| \bm{\xi}_{\bm{\phi}}^{\mathbf{w}}(t) \|_{\infty} \leq \mathfrak{r} + \kappa(\| \bm{\phi}\|_{\infty})\mathrm{e}^{-\gamma t}, \ \forall t\geq 0 
    \end{equation*}
    holds  for all $\bm{\phi}\in \mathcal{C}\{ [-r, 0], \mathbb{R}^n\}$ and $\| \mathbf{w}(t) \|_{\infty} \leq w_{max}, \forall t\geq 0$.
\end{definition}
In Definition \ref{def:ball_convergence}, $\gamma$ represents the rate of decay, i.e., an estimate of how 
quickly the solution of \eqref{dde:linearDDE} converges to the ball $\mathfrak{B}(\mathfrak{r})$.
Especially, when the radius $\mathfrak{r}=\mathbf{0}$, the definition of ball convergence is consistent with 
Lyapunov exponential stability \cite{Lyapunov_stability}. Moreover, 
in \cite{trinh2014new}, it  was proved that 
\oomit{a class of linear delay systems was  studied  
under perturbation.  there exists a ball to bound all state trajectories from some region of initial conditions,
and for the other part of initial conditions, all state trajectories from it converge 
exponentially within the ball, that is }
\begin{theorem}[\cite{trinh2014new}]
\label{thm:expenential}
Suppose in 
 \eqref{dde:linearDDE} $M=A+B$ is a Metzler matrix satisfying one of the properties in Proposition~ \ref{pro:Metzler}. Then, there exist positive constants 
$\beta$, $\gamma$, $\delta$, $\eta$ such that for all initial functions $\bm{\phi}$ and $\|\mathbf{w}(t) \|_{\infty} \leq w_{max},  \forall t\geq 0$  
\begin{equation*}
    \| \bm{\xi}_{\bm{\phi}}^{\mathbf{w}}(t)\|_{\infty} \leq \frac{C_{max} w_{max}}{\eta} + \beta ( \|\bm{\phi} \|_{\infty} -\frac{C_{max} w_{max}}{\delta})^+ \mathrm{e}^{-\gamma t}, \ \forall t\geq 0 
\end{equation*}
holds, where $C_{max} = \max_{i\in n} \sum_{j=1}^{m} C_{ij} $.
\end{theorem}

In Theorem \ref{thm:expenential}, based on the notion of Metzler matrix,  \eqref{dde:linearDDE}
is globally exponentially convergent to the ball $\mathfrak{B}(\frac{C_{max} w_{max}}{\eta})$
for all perturbations  $\| \mathbf{w}(t)\|_{\infty} \leq w_{max},  \forall t\geq 0$. Moreover, the size of the ball  increases as the perturbation bound
increases. 
Particularly, without perturbation by letting $\mathbf{w}(t) =0$ for all $t\in [-r, \infty)$, the  equilibrium $\mathbf{0}$
is exponentially stable. \cite{trinh2014new} also provides the way to 
obtain the constants $\beta,\gamma,\delta,\eta$ in Theorem \ref{thm:expenential}, which can be sketched as: let  $\bm{\zeta} > \mathbf{0}$ with 
$\| \bm{\zeta} \|_{\infty}=1$ and $ M \bm{\zeta} < \mathbf{0}$,
then 
	 $\beta=(\bm{\zeta})^{-1}_{min}$,
	 $\eta=(-M\bm{\zeta})_{min}$, 
	 $\delta=\eta (\bm{\zeta})^{-1}_{min}$,
	 $\gamma=\min_{i\in n} \gamma_i$, where $\gamma_i$ is the solution of the equation
	 \begin{equation*}
		 H_i(\gamma) = \gamma \zeta_i +\sum_{j=1}^n \zeta_j B_{ij}(e^{\gamma r - 1}) - \eta = 0  . 
	 \end{equation*}

\begin{algorithm}[t]
	\caption{Safe Differential Invariant Synthesis}
\begin{algorithmic}[1]
\Procedure{DInvariant}{ $\Xi(q)$, $\bm{f}_q$,  $T^*$, $\tau$,  $\bm{\rho}$, $\mathcal{S}_q$, $\mathfrak{r}_1$,  $\epsilon$}
  \State $P_0(q) \gets \Xi(q)\cap \mathcal{S}_q$; $i\gets 0$; $t\gets 0$
  \While{$t\leq  T^*$}
     \State $\mathrm{R}_{P_i(q)} \gets \emptyset$
     \State $\widehat{P_i}(q) \gets$ select a $C \in \mathrm{C}(P_i(q), \bm{\rho})$ \label{line:cover}
       \For{each $\hat{\mathbf{x}} \in \widehat{P_i}(q)$}
             \State $\mathrm{R}_{\hat{\mathbf{x}}} \gets$ SafeR($\bm{\rho}$, $\hat{x}$, $\tau$, $\mathcal{S}_q$)
             \If{$\mathrm{R}_{\hat{\mathbf{x}}}\neq \emptyset$}
             \State $\mathrm{R}_{P_i(q)} \gets \mathrm{R}_{P_i(q)} \cup \mathrm{R}_{\hat{\mathbf{x}}}$
             \EndIf
        \EndFor
        \If { $\mathrm{R}_{P_i(q)}\neq \emptyset$}
            \If{$\mathrm{R}_{P_i(q)} \subseteq P_{i}(q)\cup \mathfrak{B}(\mathfrak{r}_1+\epsilon) $} \label{line:fixpoint}
             \State\Return $P_{i}(q)\cup (\mathfrak{B}(\mathfrak{r}_1 +\epsilon) \cap \mathcal{S}_q)$
            \Else \State $P_{i+1}(q) \gets P_i(q)\cup \mathrm{R}_{P_i(q)}$
               \State $i\gets i+1$; $t\gets t+\tau$
            \EndIf
        \Else \State Break; 
        \EndIf
  \EndWhile
   \State\Return $P_{i}(q)\cup (\mathfrak{B}(\mathfrak{r}_1 +\epsilon)\cap \mathcal{S}_q)$ \label{line:tend}
 \EndProcedure
 
 \Procedure{SafeR}{$\bm{\rho}$, $\hat{x}$, $\tau$, $\mathcal{S}_q$}
        \State compute $\mathrm{R}_{\hat{\mathbf{x}}}$ over $t\in [0, \tau]$
        \If{$\mathrm{R}_{\hat{\mathbf{x}}}\subseteq \mathcal{S}_q $}
            \State \Return $\mathrm{R}_{\hat{\mathbf{x}}}$
        \ElsIf{$\mathrm{R}_{\hat{\mathbf{x}}}\cap \mathcal{S}_q \neq \emptyset 
        \wedge \bm{\rho}/2 \geq \bm{\rho}_{th}$ }
            \State  $\hat{Y} \gets$ $\mathrm{C}(\hat{\mathbf{x}}, \bm{\rho}/2)$ 
            \State $\mathrm{R}_{\hat{\mathbf{x}}} \gets \emptyset$
            \For{each $\hat{y} \in \hat{Y}$}
                \State  $R_y \gets$ SafeR($\bm{\rho}/2$, $\hat{y}$, $\tau$, $\mathcal{S}_q$ )
                \State $ \mathrm{R}_{\hat{\mathbf{x}}} \gets \mathrm{R}_{\hat{\mathbf{x}}} \cup R_y$
            \EndFor
        \Else  \State \Return $\emptyset$
        \EndIf
       
        \State\Return $\mathrm{R}_{\hat{\mathbf{x}}}$
     \label{line:forend}
 \EndProcedure
 
\end{algorithmic}
\label{alg:Invariant}
\end{algorithm}

\paragraph{\textbf{Reducing to $T$-differential invariant generation problem:}} 
According to Theorem~\ref{dde:linearDDE},  the first step of differential invariant generation can be achieved by the following theorem:
\begin{theorem}
\label{thm:reduceT}
Suppose $M=A+B$ is a Metzler matrix in 
 \eqref{dde:linearDDE} satisfying  one of the properties in Proposition~\ref{pro:Metzler}. 
  Given an initial function $\bm{\phi}$ and a disturbance $\mathbf{w}$ with 
$\| \mathbf{w}(t)\|_{\infty}\leq w_{max}, \forall t\geq 0$,  
 let $\mathfrak{r}_1=\frac{C_{max} w_{max}}{\eta} $ and $\mathfrak{r}_2 = \beta ( \|\bm{\phi} \|_{\infty} -\frac{C_{max} w_{max}}{\delta})$,  for any $\epsilon >0$,
 let 
\begin{equation*}
    T^* = \max\{ 0, \inf \{ T \mid \forall t\geq T : \mathfrak{r}_2^+ \mathrm{e}^{-\gamma t}<\epsilon \}\} \, ,
\end{equation*}
then 
$\|\bm{\xi}_{\bm{\phi}}^{\mathbf{w}}(T)\|_{\infty}- \mathfrak{r}_1 < \epsilon$ 
for any $T \geq T^*$, where  $\beta$, $\gamma$, $\delta$ and $\eta$ satisfy the condition in Theorem \ref{thm:expenential}. 
\end{theorem}
\begin{proof}
	The proof for the \emph{necessity} part is straightforward.
	For the \emph{sufficiency} part, by Theorem \ref{thm:expenential},
	$\| \bm{\xi}_{\bm{\phi}}^{\mathbf{w}}(t) \|_{\infty} \leq \mathfrak{r}_1 + \mathfrak{r}_2^+ \mathrm{e}^{-\gamma t} $
	for any $t\geq 0$, $\bm{\phi}$ and $\mathbf{w}(t)$. Moreover, $\mathfrak{r}_2^+ \mathrm{e}^{-\gamma t}$  is strictly monotonically decreasing w.r.t $t$,
	hence there exists an upper bound $T^*$ such that for any $t\geq T^*$, $\mathfrak{r}_2^+ \mathrm{e}^{-\gamma t}$ is exponentially close to the ball $\mathfrak{B}( \mathfrak{r}_1)$ within a 
	prescribed precision $\epsilon$. 
	Therefore, for the given precision $\epsilon$, for any $t\geq T^*$, all trajectories starting from $\bm{\phi}$ are exponentially convergent to the ball $\mathfrak{B}(\mathfrak{r}_1 )$.
\end{proof}

\oomit{
From Theorem \ref{thm:reduceT}, we have that any differential invariant of \eqref{dde:linearDDE} corresponds  to  one of its $T$-differential invariants, if $M=A+B$ is a Metzler matrix in 
 \eqref{dde:linearDDE}.}

\begin{lemma}
\label{lem:attractor}
Suppose in 
 \eqref{dde:linearDDE} $M=A+B$ is a Metzler matrix with one of the properties in Proposition~\ref{pro:Metzler}. Given $\epsilon>0$, the ball $\mathfrak{B}(\mathfrak{r}_1+\epsilon)$ is an attractor, i.e., any trajectory originating from a state in $\mathfrak{B}(\mathfrak{r}_1+\epsilon)$ is guaranteed to evolve into $\mathfrak{B}(\mathfrak{r}_1+\epsilon)$.
\end{lemma}

Theorem \ref{thm:reduceT} and Lemma \ref{lem:attractor} set up a sound guarantee that   synthesizing 
 differential invariant problem can be reduced to  synthesizing 
 $T$-differential invariant problem.
Now we are ready to introduce the second step of synthesizing differential invariants.

\paragraph{\textbf{Computing an over-approximation of reachable set within $T^*$:}} we adapt the  method in \cite{Reissig2016} for ODEs to compute an over-approximation of the reachable set for \eqref{dde:linearDDE} with  a \emph{growth bound} defined below.

\begin{definition}[Growth bound]
	Given  $t>0$, $\bm{\rho} \in \mathbb{R}^n_+$ and a compact set $K \subseteq I(q)$,  a growth bound is a map 
	$\Lambda: \mathbb{R}^n_{+} \times \mathbb{R}_{+} \mapsto \mathbb{R}^n_{+}$ 
	satisfying the following conditions:
	\begin{itemize}
		\item $\Lambda (\bm{\rho}, t) \geq \Lambda (\bm{\rho}' , t)$ whenever $\bm{\rho} \geq \bm{\rho}'$,
		\item given  $\bm{\phi}(t) \in  \mathcal{C}\{ [-r, 0], K \}$, then 
		 \begin{equation*}
		\sup_{\theta_1,\theta_2\in [-r,0]}	|\mathbf{x}_{t}^{\bm{\phi}}(\theta_1) - \mathbf{x}_{t}^{\bm{\phi}}(\theta_2) | \leq \Lambda (\sup_{\theta_1,\theta_2\in [-r,0]} |\bm{\phi}(\theta_1) - \bm{\phi}(\theta_2)|, t) \, ,
		 \end{equation*}
	where  $|\cdot |$ represents the element-wise absolute value.
	\end{itemize}
\end{definition}

Theorem \ref{thm:computeGrowth} below tells how to construct a specific growth bound $\Lambda(\cdot, \cdot)$.
\begin{theorem} 
 \label{thm:computeGrowth}
	Given a  $\bm{\rho} \in \mathbb{R}^n_+$,  let $t >0$, the map $\Lambda(\bm{\rho}, t)$,  defined by
	\begin{equation*}
		\Lambda (\bm{\rho}, t) =   \mathrm{e}^{L t} \bm{\rho} + \int_{0}^{t} \mathrm{e}^{L(t-s)} |B|\Lambda ( \bm{\rho}, s-r)  \ \mathrm{d}s \, ,
	\end{equation*}
	is a \emph{growth bound} of \eqref{dde:linearDDE}, 
	where  $L$ satisfies 
	\begin{equation*}
		L_{ij} \geq 
		\left\{
             \begin{array}{lr}
             A_{ij}, &  i=j  \\
             |A_{ij}|, & \text{otherwise} 
             \end{array}
\right..
	\end{equation*}
\end{theorem}
\begin{proof}
Given any states $\mathbf{x}(t),\mathbf{y}(t) \in I(q)$, 
let $\mathbf{z}(t) = \mathbf{y}(t)-\mathbf{x}(t) $.
From \eqref{dde:linearDDE}, 
    $ \dot{\mathbf{z}}(t) = \dot{\mathbf{y}}(t) - \dot{\mathbf{x}}(t)
             =  A \mathbf{z}(t) + B \mathbf{z}(t-r)$.
Hence, by Lemma 6 in \cite{Reissig2016}, we get
$$ |\mathbf{z}(t)| \leq  \mathrm{e}^{L t} \bm{\rho} + \int_{0}^{t} \mathrm{e}^{L(t-s)} |B|\mathbf{z}(s-r)  \mathrm{d}s.$$ 
\end{proof}

A \emph{hyper-rectangle} $[\![\mathbf{a},\mathbf{b}]\!]$ with $\mathbf{a},\mathbf{b}\in (\mathbb{R}\cup\{\pm\infty\})^n$ 
defines the set $\{x \in\mathbb{R}^n\mid a_i \leq x_i\leq b_i  \ \text{for} \ i\in\{1,\ldots,n\}\}$;
it is non-empty if $\mathbf{a} \leq  \mathbf{b}$ (element-wise).
For $\bm{\rho} \in \mathbb{R}^n_{+}$, we say that a hyper-rectangle $[\![ \mathbf{a},\mathbf{b}]\!]$
has the diameter $\bm{\rho}$ if $\frac{| \mathbf{b} - \mathbf{a}|}{2}  = \bm{\rho}$.
Given a set $K \in \mathbb{R}^n$, we denote by $\mathrm{C}(K,\bm{\rho})$ the set of covers of 
$K$, each of which is a cover of $K$, and consists of a set of hyper-rectangles with 
diameter $\bm{\rho}$. 

 Algorithm \ref{alg:Invariant} summarizes the second step  to construct a \emph{safe} 
  differential invariant: it repeats to compute the reachable set over time horizon $[0, T^*]$ in a forward way with step size $\tau$ (line 3-22); in each iteration, 
  it first finds a hyper-rectangle cover of the initial set, and any element in 
  the cover stands for an abstract state, that is a hyper-rectangle with diameter 
   $\bm{\rho}$ (line 5).  Then for each abstract state, \text{SafeR} is invoked to 
   compute the set of reachable states from the abstract state 
   within $\tau$ (line 6-11).
   If  the reachable set is not contained in the safe set, the abstract state will be refined, and \text{SafeR} is recursively invoked 
   until either the computed reachable set is contained in the safe set or the diameter of the abstract state is smaller 
   than the given threshold $\bm{\rho}_{th}$ (line 25-40); 
   this procedure terminates whenever  a fixed point is reached  (line \ref{line:fixpoint}) or the accumulated time is greater than $T^*$, and 
   returns the union of the computed reachable set before $T^*$ (i.e., $P_{i}(q)$) 
   and the over-approximation of the reachable set after $T^*$ (i.e., $\mathfrak{B}(\mathfrak{r}_1 +\epsilon)\cap \mathcal{S}_q$). 

\oomit{
\begin{algorithm}
    \caption{SafeR}
\begin{algorithmic}
 \Procedure{SafeR}{$\bm{\rho}$, $\hat{x}$, $\tau$, $\mathcal{S}_q$}
        \State compute $\mathrm{R}_{\hat{\mathbf{x}}}$ over $t\in [0, \tau]$
        \If{$\mathrm{R}_{\hat{\mathbf{x}}}\subseteq \mathcal{S}_q $}
            \State \Return $\mathrm{R}_{\hat{\mathbf{x}}}$
        \ElsIf{$\mathrm{R}_{\hat{\mathbf{x}}}\cap \mathcal{S}_q \neq \emptyset 
        \wedge \bm{\rho}/2 \geq \bm{\rho}_{th}$ }
            \State  $\hat{Y} \gets$ $\mathrm{C}(\hat{\mathbf{x}}, \bm{\rho}/2)$ 
            \State $\mathrm{R}_{\hat{\mathbf{x}}} \gets \emptyset$
            \For{each $\hat{y} \in \hat{Y}$}
                \State  $R_y \gets$ SafeR($\bm{\rho}'$, $\hat{y}$, $\tau$, $\mathcal{S}_q$ )
                \State $ \mathrm{R}_{\hat{\mathbf{x}}} \gets \mathrm{R}_{\hat{\mathbf{x}}} \cup R_y$
            \EndFor
        \Else  \State \Return $\emptyset$
        \EndIf
       
        \State\Return $\mathrm{R}_{\hat{\mathbf{x}}}$
     \label{line:forend}
 \EndProcedure
\end{algorithmic}
\end{algorithm} }
\begin{theorem}\label{thm:invariant_sound}
Given a delay dynamical system $(\Xi(q), \bm{f}_q, I(q))$ and a safety requirement $\mathcal{S}_q$, where $\bm{f}_q$ is with the form \eqref{dde:linearDDE} such that  $M=A+B$ is a Metzler matrix  satisfying one of the properties in Proposition~\ref{pro:Metzler}. Let $T^*$, $\epsilon$ and $\mathfrak{r}_1$ be defined by
Theorem \ref{thm:reduceT},  $\bm{\rho}$ and $\tau$ be the discretization parameter and  step size,  then Algorithm \ref{alg:Invariant}
terminates and returns a differential invariant  for \eqref{dde:linearDDE}.
\end{theorem}
\begin{proof}
\textbf{Termination}:  Obviously. 

\textbf{Soundness}: (i) If the algorithm returns the result at line 14, we have $\mathrm{R}_{P_i(q)} \subseteq P_i(q)\cup \mathfrak{B}(\mathfrak{r}_1 +\epsilon)$, then 
        \begin{align*}
        P_{i+1}(q)\cup \mathfrak{B}(\mathfrak{r}_1 +\epsilon) & =  P_{i}(q)\cup \mathrm{R}_{P_i(q)}\cup \mathfrak{B}(\mathfrak{r}_1 +\epsilon) \\
             &\subseteq  P_i(q)\cup \mathfrak{B}(\mathfrak{r}_1 +\epsilon) .
        \end{align*}
    By recursion,  $P_i(q)\cup \mathfrak{B}(\mathfrak{r}_1 +\epsilon)$ is an over-approximation of the reachable set over the infinite time horizon from the initial set for \eqref{dde:linearDDE}, i.e.,  $P_i(q)\cup (\mathfrak{B}(\mathfrak{r}_1 +\epsilon)\cap \mathcal{S}_q )$  is a safe differential invariant of  \eqref{dde:linearDDE}.  (ii) If the algorithm terminates at line 23,  evidently $P_i(q)$ is an over-approximation of the reachable set over time $[0, T^*]$ from the initial set of \eqref{dde:linearDDE}. By Theorem \ref{thm:reduceT} and Lemma \ref{lem:attractor}, $P_i(q)\cup (\mathfrak{B}(\mathfrak{r}_1 +\epsilon)\cap \mathcal{S}_q )$  is a  safe differential invariant for \eqref{dde:linearDDE}.
\end{proof}


\subsection{Nonlinear Systems}\label{subsec:nonlinear}
In this subsection, we generalize the two-step method in Section \ref{subsec:linear} for  nonlinear systems by means of linearization techniques. 
\oomit{and dynamics estimation. 
To the end, we need to change the first step of the differential invariant generation, which will be based on a new theory, while the second step to compute growth bound is the same as in the linear case. 
Therefore, in the following, we emphasize on how to reduce the $\infty$-invariant generation problem to the $T$-invariant generation problem for nonlinear systems. }

For simplifying the presentation,  we first consider the form of DDE \eqref{dde:DDE} with one single delay, i.e., 
\begin{equation}
    \label{dde:oneDDE}
    \left\{  
                 \begin{array}{lr}  
                 \begin{aligned}
                      \dot{\mathbf{x}}(t)= \bm{f}(\mathbf{x}(t), \mathbf{x}(t-r),\mathbf{w}(t)),
                 \end{aligned}
                 & t\in [0,\infty) \\ 
                 \mathbf{x}(t)= \bm{\phi}(t) , & t\in [-r,0]
                 \end{array}  
    \right..
\end{equation}
Let 
\begin{equation*}
    A= \left.\frac{\partial \bm{f}}{\partial \mathbf{x}(t)}\right|_{(\mathbf{0},\mathbf{0})} \text{and} \
    B = \left.\frac{\partial \bm{f}}{\partial \mathbf{x}(t-r)} \right|_{(\mathbf{0},\mathbf{0})}
\end{equation*}
 be the Jacobian matrices of DDE \eqref{dde:oneDDE} with respect to $\mathbf{x}(t)$ and $\mathbf{x}(t-r)$, evaluated at the origin $(\mathbf{0},\mathbf{0})$, respectively.
Thus, we can linearize DDE \eqref{dde:oneDDE} as 
\begin{equation}
\label{dde:linearization_higher}
 \left\{  
                 \begin{array}{lr}  
                 \begin{aligned}
                      \dot{\mathbf{x}}(t)= A\mathbf{x}(t) + B \mathbf{x}(t-r) + C \mathbf{w}(t) \\
                      + \mathbf{g}(\mathbf{x}(t),\mathbf{x}(t-r)),  
                 \end{aligned}
                 &  t\in [0,\infty) \\ 
                 \mathbf{x}(t)= \bm{\phi}(t) , &  t\in [-r,0]
                 \end{array}  
    \right. \, , 
\end{equation}
where $\mathbf{g}(\cdot,\cdot)$ is the higher-order term, which is  very closed to zero when  $\mathbf{x}$ is sufficiently close to the equilibrium. 
By dropping the higher-order term in \eqref{dde:linearization_higher}, we can obtain the approximation of \eqref{dde:oneDDE}, which is exactly the same linear system specified in  \eqref{dde:linearDDE}.

\begin{definition}[Local ball-convergence]
 Given a $\mathfrak{r} >0$,  \eqref{dde:oneDDE} is called \emph{locally exponentially convergent} within the ball
    $\mathfrak{B}(\mathfrak{r})$, if there exist constant $\gamma > 0$, $\iota >0$ and a non-decreasing function $\kappa(\cdot)$
    such that for all  $\| \mathbf{w}(t) \|_{\infty} \leq w_{max}$
    \begin{equation*}
       \| \bm{\phi}(t)  \|_{\infty} \leq \iota \implies \| \bm{\xi}_{\bm{\phi}}^{\mathbf{w}}(t) \|_{\infty} \leq \mathfrak{r} + \kappa(\| \bm{\phi}\|_{\infty})\mathrm{e}^{-\gamma t},  \ \forall t\geq 0
    \end{equation*}
    holds. 
\end{definition}

\begin{theorem}
\label{thm:expenential_nonlinear}
Suppose that $M=A+B$ is a Metzler matrix in \eqref{dde:linearization_higher} satisfying  one of two properties in Proposition~\ref{pro:Metzler}, then there exist positive constants $\iota$,
$\beta$, $\gamma$, $\delta$ and $\eta$ such that  for all  $\|\mathbf{w}(t)\|_{\infty} \leq w_{max}$ 
\begin{align*}
          &\| \bm{\phi}(t) \|_{\infty} \leq \iota \implies 
      & \| \bm{\xi}_{\bm{\phi}}^{\mathbf{w}}(t) \|_{\infty} \leq \frac{\mathcal{G}}{\eta} + 
      \beta ( \|\bm{\phi} \|_{\infty} -\frac{\mathcal{G}}{\delta})^+ \mathrm{e}^{-\gamma t}, \ \forall t\geq 0
\end{align*}
holds, where $C_{max} = \max_{i\in n} \{ \sum_{j=1}^{m} C_{ij} \}  $. 
\end{theorem}

\begin{proof}
Let $\mathcal{G} = C_{max} w_{max}+ g_{max}$, where $\|\mathbf{g}(\mathbf{x}(t),\mathbf{x}(t-r)) \|_{\infty} \leq g_{max}$,
and $\iota \leq  \frac{\mathcal{G}}{\eta} + 
      \beta ( \|\bm{\phi} \|_{\infty} -\frac{\mathcal{G}}{\delta})^+ $, 
      then it can be proved similar to that of Theorem \ref{thm:expenential}. 
\end{proof}

\oomit{ To find $g_{max}$, we set the first function $H_1(g_{max})$ with  
\begin{equation}
\label{eq:1}
    \| \bm{\xi}_{\bm{\phi}}^{\mathbf{w}}(t) \|_{\infty} \leq H_1(g_{max}), 
\end{equation}
where $H_1(g_{max}) = \frac{\mathcal{G}}{\eta} + 
      \beta ( \|\bm{\phi} \|_{\infty} -\frac{\mathcal{G}}{\delta})^+ $.
For $\mathbf{g}(\mathbf{x}(t),\mathbf{x}(t-r))$, applying \eqref{eq:1}: 
\begin{equation*}
    \|\mathbf{x}(t)\|_{\infty} , \, \|\mathbf{x}(t-r)\|_{\infty} \leq H_1(g_{max}),
\end{equation*}
We have the second function $H_2(g_{max})$:
\begin{equation*}
\label{eq:2}
    \| \mathbf{g}(\mathbf{x}(t),\mathbf{x}(t-r)) \|_{\infty} \leq H_2(g_{max}).
\end{equation*}
$g_{max}$ can be obtained  by solving the inequality: $H_2(g_{max}) \leq g_{max}$.
}
Similarly, Theorem \ref{thm:reduceT_nonlinear} says that 
 the differential invariant generation problem for nonlinear DDEs can be 
 equivalently reduced to to the $T$-invariant generation problem.
\begin{theorem}
\label{thm:reduceT_nonlinear}
Given an initial function $\bm{\phi}$ and a disturbance $\mathbf{w}$ with $\|\mathbf{w}(t)\|_{\infty}\leq w_{max}, \forall t\geq 0$, 
for \eqref{dde:DDE}, 
suppose that the positive constants $\iota$,  $\beta$, $\gamma$, $\delta$, $\eta$ and $g_{max}$ satisfy the condition in Theorem \ref{thm:expenential_nonlinear}, let $\mathfrak{r}_1=\frac{\mathcal{G}}{\eta} $ and $\mathfrak{r}_2 = \beta ( \|\bm{\phi} \|_{\infty} -\frac{\mathcal{G}}{\delta})$, and for any $\epsilon >0$, let 
   $T^* = \max\{ 0, \inf \{ T \mid \forall t\geq T : \mathfrak{r}_2^+ \mathrm{e}^{-\gamma t}<\epsilon \}\}$, 
then for any $\| \bm{\phi}(t)  \|_{\infty} \leq \iota$ and any $T \geq T^*$ it follows 
$\|\bm{\xi}_{\bm{\phi}}^{\mathbf{w}}(T)\|_{\infty}- \mathfrak{r}_1 < \epsilon$. That is,  
a differential invariant of \eqref{dde:DDE} exactly corresponds to 
one of its $T$-differential invariant. 
\end{theorem}
\begin{proof}
Similar to the proof of Theorem \ref{thm:reduceT}.
\end{proof}

\begin{remark}
 The fact that Theorem \ref{thm:reduceT_nonlinear} holds with the condition $\| \bm{\phi}(t)  \|_{\infty}  \leq \iota$ implies the locality of linearization. 
 Moreover, in order to alleviate  conservativeness of linearization, we need to 
 compute a tighter  parameter $g_{max}$, which is used to bound the high-order terms 
 discarded during linearization.
 \oomit{Nevertheless, for any $\bm{\phi'}$ such that $\| \bm{\phi'}(t)  \|_{\infty} > \iota$, if the reachable set from $\bm{\phi'}$ can be contained in $\mathfrak{B}(\iota)$ during  time $[T'-r,T']$, then it is sufficient to take $T'+T^*$ as the bounded time of invariant generation problem.}
 
 Note that the above discussion can be straightforwardly  extended to 
 DDEs \eqref{dde:DDE} with multiple delays by just letting  $M=A+ \sum_1^k B_i$.
\end{remark}

\oomit{
\begin{remark}
Note that our method  is based on ball-convergence analysis of linear DDEs using  Metzler matrices and linearization, different from the method proposed in \cite{feng2019taming} based on stability analysis and linearization for safety verification of DDEs. As we discussed before, our method is more powerful, and therefore could be applied to more DDEs, which will be  indicated in the case studies.  
\end{remark} }

\section{Switching Controller Synthesis with Delays and Perturbations}\label{sec:synthesis}
In this section we present our synthesis framework based on invariant generation 
for delay hybrid systems with perturbations modelled by dHA. 
\oomit{For each mode $q$, a differential invariant $I^*(q)$ can be computed by employing the results in Section \ref{sec:invariant}.
For each edge $e=(q, q')\in E$, 
we first strengthen its guard condition 
without considering  delay $D(e)$ based on the computed 
differential invariants in $q$, denote the strengthened guard condition 
$\widetilde{G}(e)$. 
Then,  
we reduce the problem to 
searching for a set $G^*(e)$ such that all trajectories  starting from it 
can reach to $\widetilde{G}(e)$ by taking delay $D(e)$ into consideration. }

\begin{algorithm}[t]
	\caption{Backward Reachable Set Computation}
\begin{algorithmic}[1]
 \Procedure{BackReach}{$\widetilde{G}(e)$,   $D(e)$,   $I^*(q)$, $\bm{\rho}$, $\tau$}
  \State $G^*(e) \gets \emptyset$
  \State $\widehat{I}^*(q) \leftarrow  \mathrm{C}(I^*(q), \bm{\rho})$
  \State $\mathbf{d} \gets | \sup_{\mathbf{x}\in I^*(q), \|\mathbf{w}(t)\|_{\infty}\leq w_{max}} \bm{f} | \cdot D(e)$
  \For {each $\hat{\mathbf{x}} \in \widetilde{G}(e) \uplus \mathbf{d}$}
	 \State compute $\mathrm{R}_{\hat{\mathbf{x}}}(t) $ for $t\in [0, D(e)]$ with step size  $\tau$
	 \If{$\mathrm{R}_{\hat{\mathbf{x}}}(D(e))  \subseteq \widetilde{G}(e)\bigwedge \mathrm{R}_{\hat{\mathbf{x}}}(t)\subseteq I^*(q),$ \Statex\ $\qquad\qquad\qquad\qquad \quad\qquad  \forall t\in [0, D(e)]$}
		\State $G^*(e) \leftarrow G^*(e) \cup \hat{\mathbf{x}} $
		\EndIf
	 \If {$\mathrm{R}_{\hat{\mathbf{x}}}(D(e)) \cap \widetilde{G}(e) \neq \emptyset \bigwedge \mathrm{R}_{\hat{\mathbf{x}}}(t)\subseteq I^*(q),$ \Statex\ $\qquad\qquad\qquad\qquad \quad\qquad \forall t\in [0, D(e)]$ }
		 \State refine $\hat{\mathbf{x}}$ with $\bm{\rho}' \gets \bm{\rho}/2$, ($\bm{\rho}' \geq \bm{\rho}_{th}$)
	 \EndIf
  \EndFor
  \State\Return $G^*(e)$
  \EndProcedure
\end{algorithmic}
\label{alg:backreachset}
\end{algorithm}

\subsection{Computing Guards of Discrete Jumps}\label{subsec:backsec}
In this subsection, by computing a reachable set 
from the set of states reachable to the edge 
without the jump delay backwards, we focus on how to synthesize a new guard $G^*(e)$ of each discrete 
jump $e$  in order to guarantee the safety  when taking the jump delay into consideration. 

\begin{definition}[Backward reachable set]
	\label{def:backreachset}
	For a mode $q$ of the dHA $\mathcal{H}$ : $(\Xi(q), \bm{f}_q, I^*(q))$, 
	given a target region $\widetilde{G}(e)$ and a finite time $t=D(e)$, the reachable set 
	$G^*(e)$ from  the target region $\widetilde{G}(e)$ backwards after  $t$ time units is defined as 
	\begin{equation*}
	    	G^*(e) = 
	    	\left\{ \mathbf{x}_0 
	    	   \,\middle|\,
	    	\begin{array}{cc}
	    	    &\forall \ t\in [0,D(e)] , \forall \ \mathbf{w}(t). \\
	    	    &\bm{\xi}_{\mathbf{x}_0}^{\mathbf{w}} (D(e)) \in \widetilde{G}(e) \wedge 
	    	    \bm{\xi}_{\mathbf{x}_0}^{\mathbf{w}} (t) \in I^*(q)  
	    	\end{array}
	    	\right\}.
	\end{equation*} 
\end{definition}

\oomit{
Definition \ref{def:backreachset} describes
a set $G^*(e)$  that all trajectories  starting from  it  will reach to the  target
set with $t=D(e)$ time units, while all possible reachable states stay 
within $ I^*(q)$ during the 
time interval $[0,D(e)]$, regardless of perturbations.

As the best of our knowledge, few approaches have been developed to calculate the backward reachable set for DDEs. One option is to seek for all states in the invariant 
but this may be wasteful if, e.g., only a small part of state space is able to reach the target region in the bounded time $D(e)$.
We come up with a local algorithm to search for the backward 
reachable set locally instead of a global search.}

\oomit{
\begin{remark}
In Definition \ref{def:backreachset}, we define the backward reachable set as a set of  states in the form of $n$-dimensions vectors rather than functionals
 because a guard condition represented in terms of functionals 
 is difficult to check its satisfiability; moreover, 
 computing reachable sets according to Theorem \ref{thm:computeGrowth} is much easier and more intuitive. 
 Actually, it is not hard to prove the equivalence between the two forms of the representations of states. 
\end{remark} }

Now, we present an algorithm, which is presented in Algorithm~\ref{alg:backreachset}, to under-approximate 
the backward reachable set based on discretization in a symbolic way. The basic idea is: 
Given a discretization step size $\bm{\rho} \in \mathbb{R}^n_{+}$, let 
$\widehat{I}^*(q)$ be in  $\mathrm{C}(I^*(q),\bm{\rho})$, and 
$\mathbf{d} \in \mathbb{R}^n$ be 
$| \sup_{\mathbf{x}\in I^*(q), \|\mathbf{w}(t)\|_{\infty}\leq w_{max}} \bm{f} |  \cdot D(e)$, 
standing for the maximal distance following the DDE from 
$I^*(q)$ within the time delay $D(e)$ subject to any disturbance.
So, a necessary condition that an abstract state in $\widehat{I}^*(q)$
  can reach $\widetilde{G}$ within $D(e)$ is that 
  the distance from the state to $\widetilde{G}$  is less than or equal to $\mathbf{d}$, i.e., in 
  the following set
\begin{equation*}
    \widetilde{G}(e) \uplus \mathbf{d} = 
    \left\{ \hat{\mathbf{x}} \in \widehat{I}^*(q) \,\middle |\,
    \begin{array}{cc}
         \hat{\mathbf{x}} \in \widetilde{G}(e) \vee 
        \exists\ \hat{\mathbf{x}}'\in \widetilde{G}(e): \\
	    | ctr(\hat{\mathbf{x}}) - ctr(\hat{\mathbf{x}}') | \leq \mathbf{d}+4\bm{\rho} 
    \end{array}
	\right\}\, , 
\end{equation*}
where $ctr(\hat{\mathbf{x}})$ is the \emph{center} of the abstract  state 
$\hat{\mathbf{x}}$,  standing for the hyper-rectangle $[\![\mathbf{a},\mathbf{b}]\!]$, 
i.e., the point $(\frac{1}{2}(b_1-a_1),  \ldots, \frac{1}{2}(b_n - a_n))$.
Obviously,  all trajectories starting from  the set $\widehat{I}^*(q) \setminus (\widetilde{G}(e) \uplus \mathbf{d})$ are impossible to reach to $\widetilde{G}(e)$ within  $D(e)$.
Therefore, we only need to consider the set $\widetilde{G}(e) \uplus\mathbf{d}$.
For each abstract state $\hat{\mathbf{x}} \in \widetilde{G}(e) \uplus \mathbf{d}$, the over-approximation of 
the backward reachable set $\mathrm{R}_{\hat{\mathbf{x}}}$ is calculated by checking 
whether it keeps  
$I^*(q)$ satisfied over $[0, D(e)]$  and all elements of $\mathrm{R}_{\hat{\mathbf{x}}}(D(e))$ should satisfy 
$\widetilde{G}(e)$. If the answer is \emph{yes}, then it is done; otherwise, if  some of reachable states in $\mathrm{R}_{\hat{\mathbf{x}}}(D(e))$ satisfy $\widetilde{G}(e)$, then refine the abstract state $\hat{\mathbf{x}}$ with a smaller discretization parameter $\bm{\rho}'$, say  $\bm{\rho}'=\bm{\rho}/2$. 
Repeat the above procedure until  all abstract states in the set $\widetilde{G}(e) \uplus \mathbf{d}$ are done.

\subsection{Switching Controller Synthesis}\label{sec:framework}
\begin{algorithm}[t]
\caption{Switching Controller Synthesis}
 \begin{algorithmic}[1]
 \Require ~~ 
 $\mathcal{H}=(Q, X, U, I, \Xi , F , E, D, G, R )$, $\mathcal{S}$, $\bm{\rho}$, $\tau$, $\{T^*_{q}\mid  q\in Q \}$,  \ $ \{\mathfrak{r}^q_{1}\mid  q\in Q\}$, $\{\epsilon_{q} \mid   q\in Q \}$
    \State $ K_0 \gets \Xi$;  $I_0\gets \emptyset$; $\textit{flag} \gets \textbf{true}$; $\widetilde{G}_0\gets \emptyset$; $n\gets 0$      \label{line:initial}
    \While{ $\textit{flag}$}
    	 \State $n\gets n+1$
    	\For{each $q\in Q$}
    	    	\State $K_n(q) \gets K_{n-1}(q) \cup \{ 
    	                    \bm{\phi} \mid
    	                          \exists e=(q', q)\in E, \exists t>0,$
    	                      \Statex $\qquad \quad \exists \theta \in [-r_k^q,0].\ \bm{\phi}=R(e, \mathbf{x}_{t}^{\bm{\phi}}(\cdot)) \wedge  \mathbf{x}_{t}^{\bm{\phi}}(\theta)\in \widetilde{G}_{n-1}(e)
    	                     \} $ \label{line:kn}
    		\State $I_n(q)\gets$ \textbf{DInvariant}$( K_n(q) , \bm{f}_q, T^*_q, \tau, \bm{\rho},\mathcal{S}_q,\mathfrak{r}^q_{1} , \epsilon_q)$ \label{line:In}
    	    \State $\Xi(q) \leftarrow \Xi(q) \cap \mathcal{S}_q$ \label{line:xi}
    	\EndFor
    	\For{ each $e=(q,q')\in E$}
    	    \State $\widetilde{G}_n(e) \leftarrow  G(e) \cap I_n(q) \cap  
    	           \{ \mathbf{x}_{t}^{\bm{\phi}}(\theta)\in I_n(q) \mid \exists t>0, $
    	            \Statex $\qquad\qquad\qquad\qquad\qquad \forall \theta \in [-r_k^q,0].\ R(e,\mathbf{x}_{t}^{\bm{\phi}}(\cdot))\in U_{q'} \}$ \label{line:g}
    		\State $G^*_n(e)\gets$ \textbf{BackReach}($\widetilde{G}_n(e)$,   $D(e)$,   $I_n(q)$, $\bm{\rho},\tau$) \label{line:G} 
    	\EndFor
    	\If{$I_n == I_{n-1}$} \label{line:fixed}
    	    \State $\textit{flag}\gets \textbf{false}$
    	 \EndIf
    	\EndWhile 
    	\State $\Xi^* \leftarrow \{ \Xi(q) \mid q\in Q\}$
    	\State $I^* \leftarrow \{I_n(q)\mid q\in Q\}$
    	\State $U^* \leftarrow \{ \mathbf{x}_{t}^{\bm{\phi}}(\cdot)\in U \mid \exists q\in Q,\ \mathbf{x}_{t}^{\bm{\phi}}(\theta)\in I^*(q), \forall \theta \in [-r_k^q,0] \}$ \label{line:u}
    	\State $G^* \leftarrow \{(e,G^*_n(e)) \mid e\in E \}$
	\If {$\forall e\in E, G^*(e) \neq \emptyset$ }
	   \State \Return $\mathcal{H^*}\leftarrow (Q, X, U^*, I^*, \Xi^* , F , E,D, G^*, R )$ \label{line:return}
	\EndIf
 \end{algorithmic}
\label{alg:switch_controller}
\end{algorithm}

To present our approach on switching controller synthesis, 
we need to introduce the notion of \emph{global invariant}, which can be formally defined  as follows. \oomit{framework (Algorithm \ref{alg:switch_controller}) is based on the \emph{global invariant} on all modes $Q$, to obtain a non-trivial switching controller if it exists. We define the global invariant here:}
\begin{definition}[Global invariant]
    Given a dHA $\mathcal{H}$, $I^* = \cup_{q\in Q} I^*(q)$ is \emph{global invariant} of $\mathcal{H}$,
    if $I^*$ satisfies the following conditions:
    \begin{enumerate}
        \item[(c1)] for each $q\in Q$, the set $I^*(q)$ is a differential invariant of $(\Xi(q), \bm{f}_q, I(q))$,
        \item[(c2)\label{c2}] for each $e=(q, q')\in E$, if $\bm{\xi}^{\mathbf{w}}_{\bm{\phi}} (t) \in  G^*(e)$,
         then  
              \[ \forall \theta\in [t'-r^{q'}_k, t'],\ \bm{\phi}'(\theta)\in I^*(q'), \]
         where $\bm{\phi}'(\cdot) = R(e, \mathbf{x}_{t'}^{\bm{\phi}}(\cdot) )$ and $t'=t+D(e)$.
       \end{enumerate}
    \label{def:inductive_invariant}
\end{definition}

\begin{figure}[!htbp]
    \centering
    \includegraphics[width=4.5cm, height=3cm]{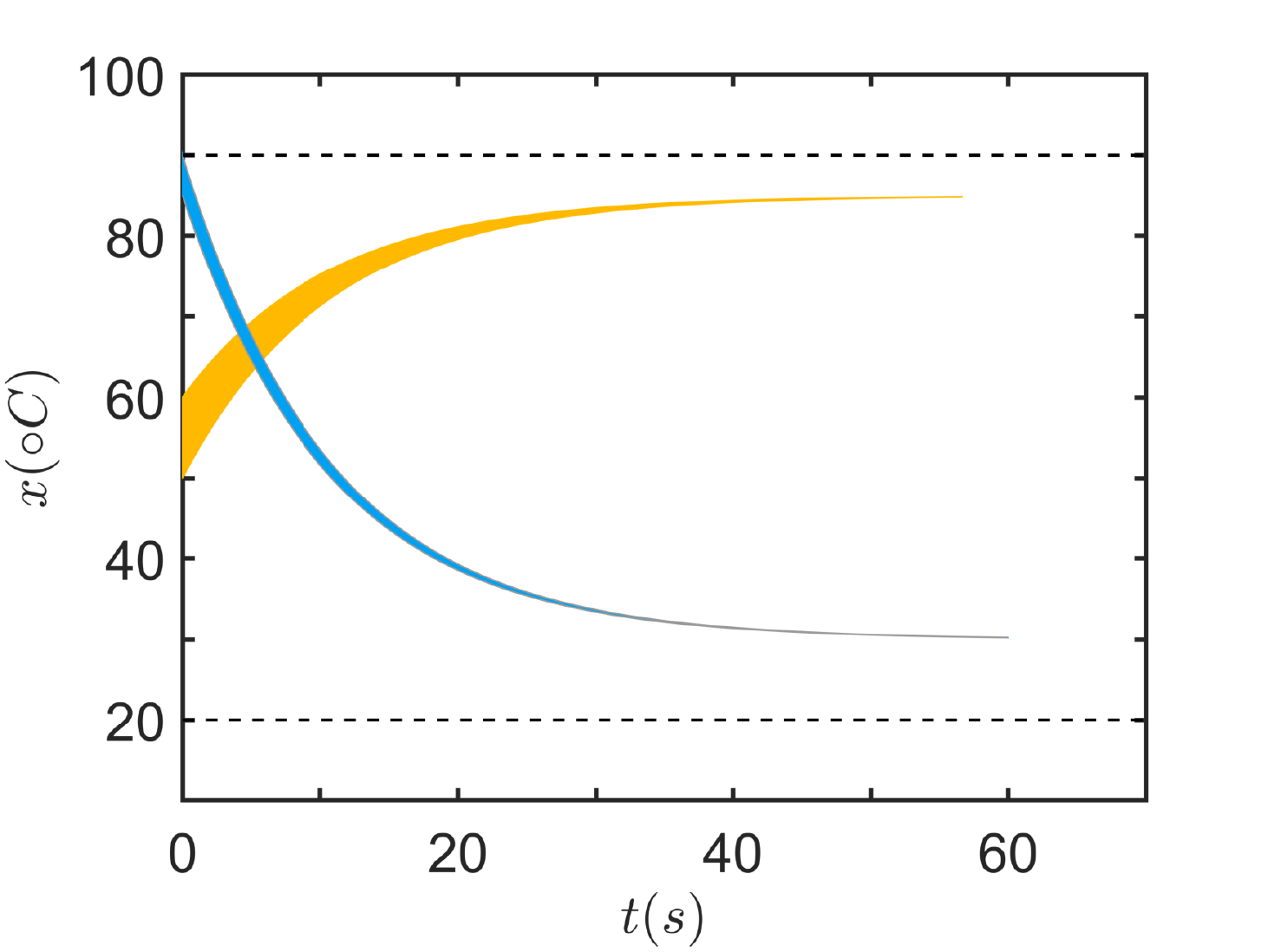}
    \caption{ The over-approximate reachable sets for two modes of the heating system. Black dashed lines denote the safety set.}  
     \label{fig:heatingreach}
 \end{figure}

Algorithm \ref{alg:switch_controller} presents a procedure to compute a global invariant 
repeatedly until the safety requirement can be guaranteed by a computed global invariant 
(when $\textit{flag}$ holds, line 2-16), then a switching controller solving Problem 1 can be  defined by the global invariant (line 17-23). In each iteration, for each mode (line 4-8), 
we compute a new mode invariant (line 5), a new differential invariant that can guarantee 
the safety requirement (line 6) by invoking Algorithm \ref{alg:Invariant} (line \ref{line:In}), and a new initial condition satisfying the safety requirement (line 7); for each discrete transition (line 9-12), we compute a new guard condition without considering the discrete delay (line 10), and then a 
new guard condition considering the discrete delay by calling Algorithm \ref{alg:backreachset} (line \ref{line:G});
then we test whether a global invariant that can guarantee the safety requirement is achieved (line 13-15). 
\oomit{
outlines an iterative synthesis framework for solving the switching controller problem of delay hybrid systems. The set $K_n$  represents the initial set using to  construct the invariants $I_n$  of dHA $\mathcal{H}$, initially, $K_0=\Xi$ and $I_0=\emptyset$ (line \ref{line:initial}).
In the $n$-th iteration, for each mode $q\in Q$, we generate its  differential invariant $I_n(q)$ by the algorithm \ref{alg:Invariant} (line \ref{line:In}). 
Based on the differential invariant, we refine the initial state set $\Xi(q)$, invariant $I_n(q)$  by intersecting with safe requirement or new invariant (lines \ref{line:xi}-\ref{line:I_refine}). 
By Definition \ref{def:inductive_invariant}, for the edge $e=(q,q')$,  a guard condition without jump delays $\widetilde{G}_n(e)$ is obtained by guaranteeing the reset functions are contained in $I_n(q')$ (line \ref{line:g}).
Furthermore, in order to deal with the delays in discrete transitions,
\emph{delay} guard condition $G^*_n(e)$ is calculated by Algorithm \ref{alg:backreachset} (line \ref{line:G}). 
Once a fixed point is found (line \ref{line:fixed}), the iteration ends up. The algorithm assigns the new results to a new dHA $\mathcal{H^*}$ and returns it (line \ref{line:return}). Otherwise, the algorithm goes into next  iteration. 
After updating $K_{n+1}$ with new guard functions from previous iteration (line \ref{line:kn}),  continue the process as above.}

The soundness of our approach is guaranteed by the following theorem. 
\begin{theorem}[Soundness] \label{thm:framework_sound}
Given a hybrid automaton $\mathcal{H}=(Q, X, U,I, \Xi , F , E, D, G, R )$ and its safety property $\mathcal{S}$, 
a dHA $\mathcal{H}^*=(Q, X, U^*, I^*, \Xi^*, F, E,D, G^*, R )$ constructed by Algorithm \ref{alg:switch_controller}  fulfills the three requirements (r1)-r(3) in Problem \ref{prob:synthesis}.
\end{theorem}
\begin{proof}
We first prove that $I^*$ is a safe global invariant of $\mathcal{H}^*$ if  Algorithm \ref{alg:switch_controller} terminates and returns 
$\mathcal{H}^*=(Q, X, U^*, I^*, \Xi^* , F , E,$ $D, G^*, R )$, i.e., the conditions (c1) and (c2) in Definition \ref{def:inductive_invariant} with restriction of safety requirement $\mathcal{S}$ hold.
From line \ref{line:In} in Algorithm \ref{alg:switch_controller}, Definition \ref{def:diff_invariant} and the soundness of Algorithm \ref{alg:Invariant}, we have $I^*(q)$ is a safe differential invariant of $(\Xi^*(q),f_q,I(q)$, then (c1) holds.
Let $e=(q,q')\in E$, and $\bm{\xi}_{\bm{\phi}}^\mathbf{w}(t)\in G^*(e)$. 
From line \ref{line:kn}, \ref{line:In} in Algorithm \ref{alg:switch_controller}, we have
\begin{align*}
\left\{  \bm{\phi} \,\middle |\,
\begin{array}{cc}
     \exists e=(q', q)\in E, \exists t>0, \exists \theta \in [-r_k^q,0].\ \\
     \bm{\phi}=R(e, \mathbf{x}_{t}^{\bm{\phi}}(\cdot)) \wedge  \mathbf{x}_{t}^{\bm{\phi}}(\theta)\in \widetilde{G}(e)
\end{array}
\right\} \subseteq I^*(q).
\end{align*}
From line \ref{line:G}, it follows  
$$G^*_n(e) =  \textbf{BackReach}(\widetilde{G}_n(e),D(e), I_n(q), \bm{\rho},\tau),$$
which implies (c2) holds. 
Now, we prove  that (r1), (r2) and (r3) in Problem \ref{prob:synthesis} are satisfied.
Since each $I_n(q)$ is calculated by Algorithm \ref{alg:Invariant}, which can guarantee
 $I_n(q)$ is safe, thus $\mathcal{H}^*$ is safe, i.e., (r1) holds.
In Algorithm \ref{alg:switch_controller}, line \ref{line:xi} makes $\Xi^* \subseteq \Xi \cap \mathcal{S}$, line \ref{line:In} makes $I^* \subseteq I$.
From line \ref{line:u}, and $I^* \subseteq I$, it follows  $U^* \subseteq U$.
For any $e\in E$, as there exists $\theta \in [-r^q_{k},0]$ such that $\mathbf{x}^{\bm{\phi}}_{t}(\theta)\in G^*(e)$, 
hence $\mathbf{x}^{\bm{\phi}}_{t+D(e)}(\theta)\in \widetilde{G}(e)$.
From line \ref{line:g} and \ref{line:G}, it follows  $\mathbf{x}^{\bm{\phi}}_{t+D(e)}(\theta) \in G(e)\cap I^*(q)$. Thus, (r2) holds.
Clearly, $I^*$ contains all safe trajectories of $\mathcal{H}$, so if $\mathcal{H}$ is non-blocking with respect to  the safe requirement  $\mathcal{S}$, then $\mathcal{H}^*$ is also non-blocking, i.e., (r3) holds. 
\end{proof}

\begin{example} \label{ex:heating}
We continue to consider the heating system example. 
Let $K_1=0.25$, $K_2=0.15$,  $h=32$,  $w_1=0.5$, and $w_2=3$ for the dHA of the heating system in Example \ref{exm:heating_automaton}. 
For mode $q_1$, $M_{q_1}=-0.1$ is trivially a Metzler matrix. 
Applying Theorem \ref{thm:reduceT}, we have $T^*_{q_1} = 56.567$s. The same procedure applies  to mode $q_2$, we have $T^*_{q_2}= 60.043$s. By Algorithm \ref{alg:switch_controller},  we obtain differential invariants $I^*(q_1) = \{x \mid 30 \leq x\leq 84.91\}$ and  $I^*(q_2) = \{x \mid 30.2056 \leq x\leq 90\}$. 
Also, strengthened guarded conditions on $e_1$ and $e_2$ can be easily computed as $G^*(e_1) =\{x \mid 30 \leq x\leq 84.30\}$ and $G^*(e_2) =\{x \mid 34.5 \leq x\leq 90\}$. 
The over-approximation of the reachable sets from the initial sets in the 
two modes respectively  are displayed in Figure \ref{fig:heatingreach}. 
\end{example}


\begin{table*}[!htbp]
\centering
\begin{minipage}[b]{0.4\textwidth}
 \begin{tabular}{|c|c|c|c|c|c |c|}
    \hline
      Mode   & $\epsilon$ &$\zeta$ & $\beta$ & $\eta$ & $\gamma$ & $\delta$ \\ 
     \hline
      $q_1$  & 0.001  &  $\begin{bmatrix}
            1  \\ 
            1  
        \end{bmatrix}$     & 1    & 12.58     &5.1642  & 12.58 \\
             \hline
              $q_2$ &  0.001  &  $\begin{bmatrix}
            1  \\ 
            1  
        \end{bmatrix}$   &1   &  24.66   &4.2270   &24.66   \\
             \hline
        \end{tabular}
\captionsetup{justification=centering}
  \caption{The value of parameters in Section~\ref{case:1}}
    \label{tab:lowpass}
\end{minipage}
 \hspace{-4mm}
 \begin{minipage}[b]{0.5\textwidth}
    \begin{tabular}{|c|c|c|c|c|c |c|c|c|c|}
    \hline
      Mode   & $\epsilon$ &$\zeta$ & $\beta$ & $\eta$ & $\gamma$ & $\delta$   & $g_{max}$  & $\mathcal{G} $ & $\iota$   \\ 
     \hline
      $q_1$  & $10^{-4}$ &    $\begin{bmatrix}
            1  \\ 
            1  
        \end{bmatrix}$     & 1   & 0.8    &0.626  & 0.8 & 0.008    & 0.078   & 0.2  \\
             \hline
              $q_2$ &  $10^{-4}$   & $\begin{bmatrix}
            1  \\ 
          1  
        \end{bmatrix}$   &1   &  1.85  &0.88  & 1.85 & 0.0046   & 0.0746  &0.2 \\
             \hline
        \end{tabular}
 \captionsetup{justification=centering}
  \caption{The value of parameters in Section~\ref{case:2}} 
    \label{tab:prepre}
 \end{minipage}
 \end{table*}

\section{Experimental Results}\label{sec:expri}
 We implement our algorithms \footnote{Available at  \url{https://github.com/YunjunBai/Inv\_DHA}.} in Matlab, based upon the interval data-structure in CORA \cite{Althoff2016a}.  
We adopt  the discretization  parameters  from \cite{Althoff2016a} and \cite{Ghafli2020} for the two examples, respectively. 
All experiments are  performed on an Intel(R) Core(TM) i5-8265U CPU
 (1.60GHz) with 8GB RAM. 

 \begin{figure*}[!htbp]
    \centering
    \subfigure[]{
        \label{fig:filter_a}
    \includegraphics[width=4.5cm, height=3cm]{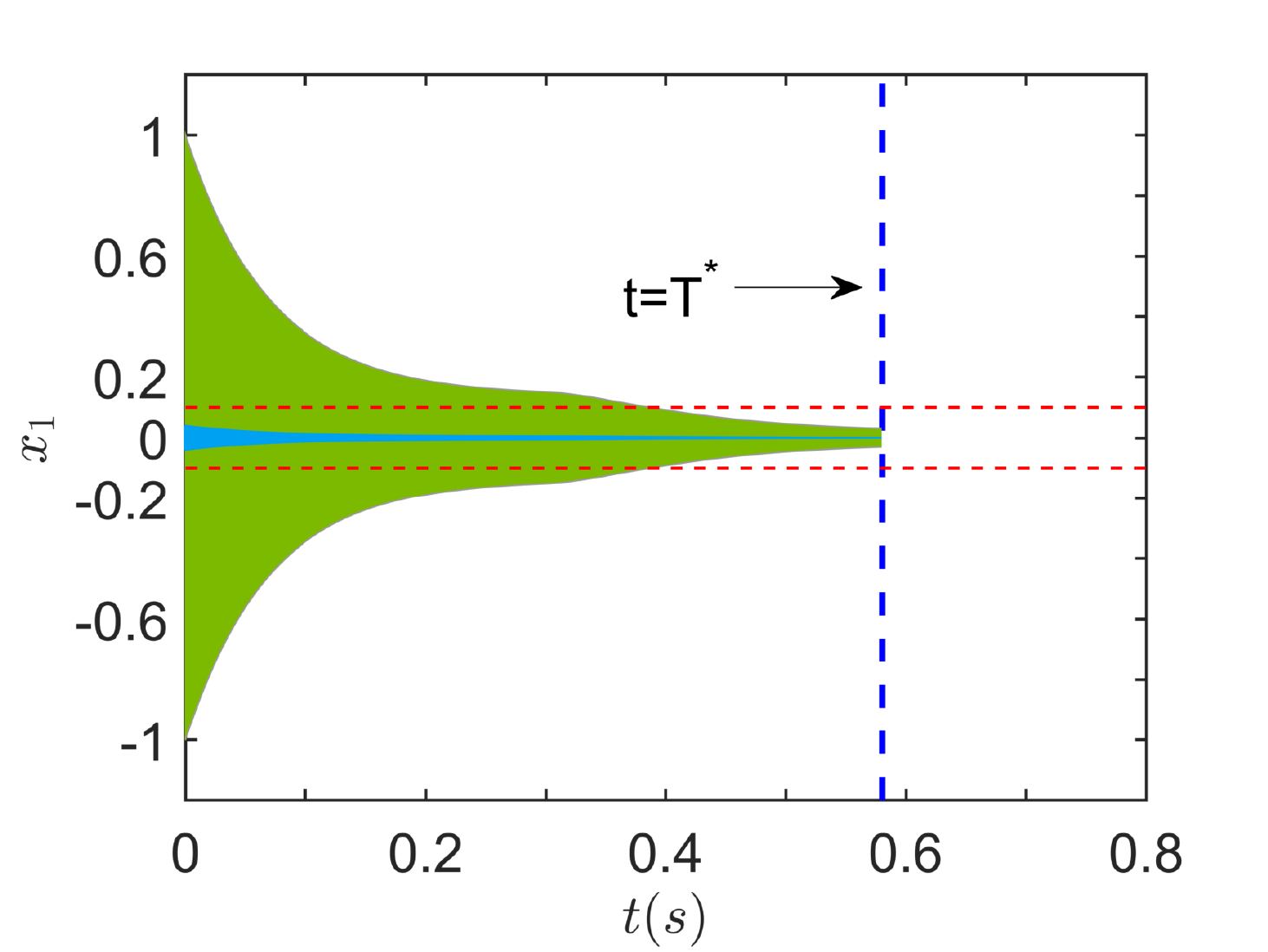}
    }
   \hspace{-7mm}
    \subfigure[]{
        \label{fig:filter_b}
    \includegraphics[width=4.5cm,height=3cm]{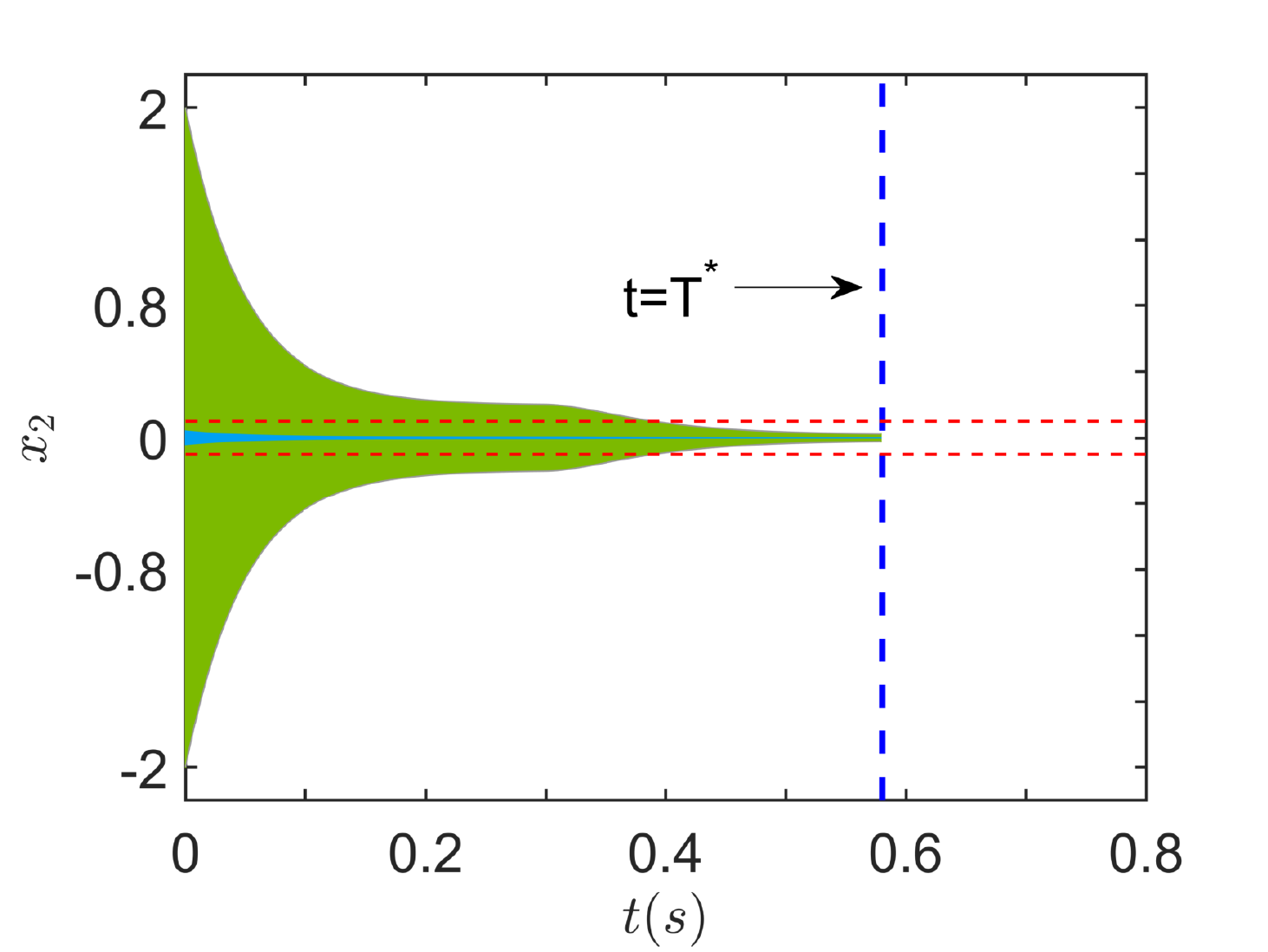}
    }
   \hspace{-7mm}
    \subfigure[ ]{
        \label{fig:filter_c}
    \includegraphics[width=4.5cm, height=3cm]{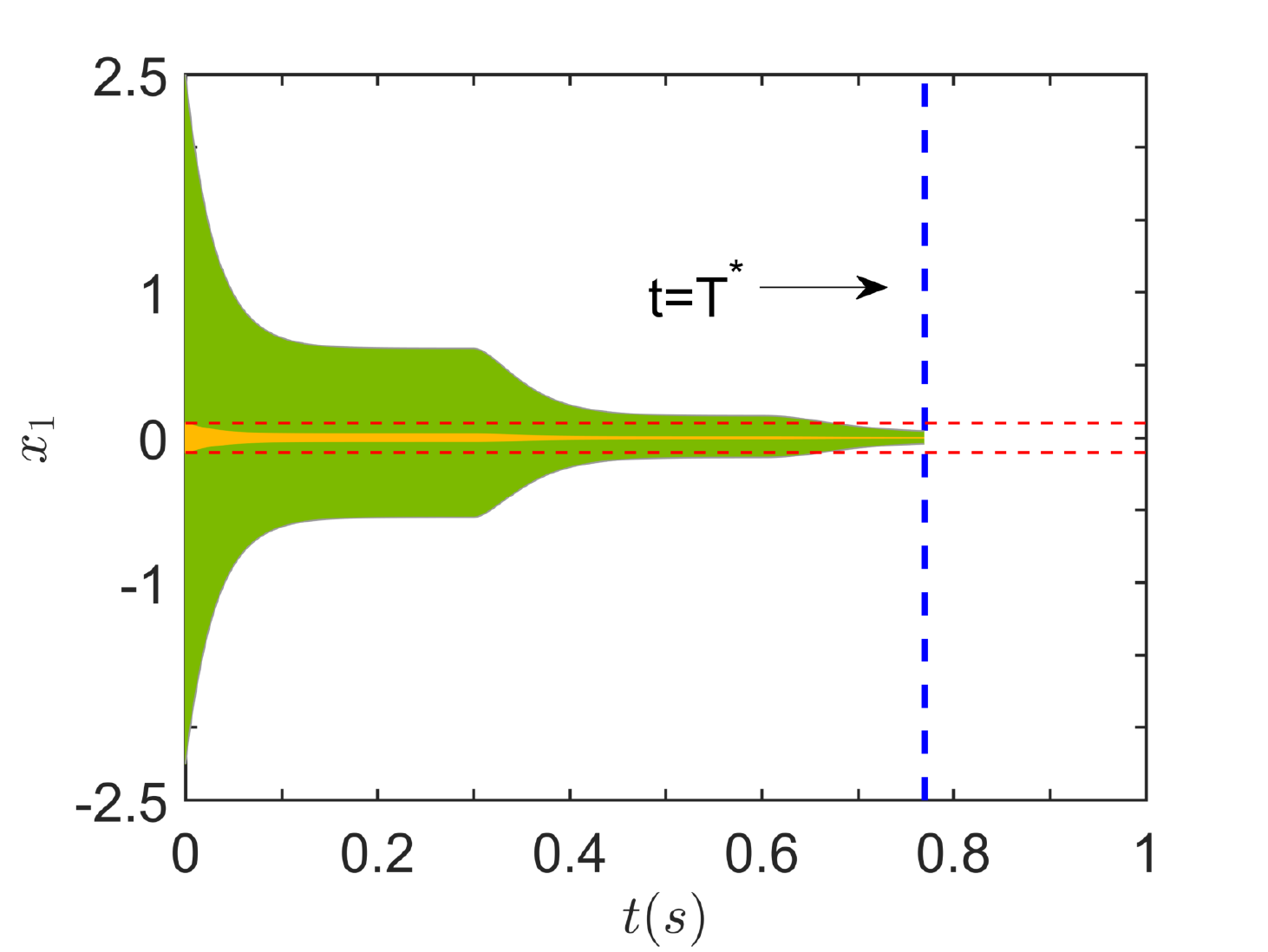}
    }
    \hspace{-7mm}
    \subfigure[]{
        \label{fig:filter_d}
    \includegraphics[width=4.5cm,height=3cm]{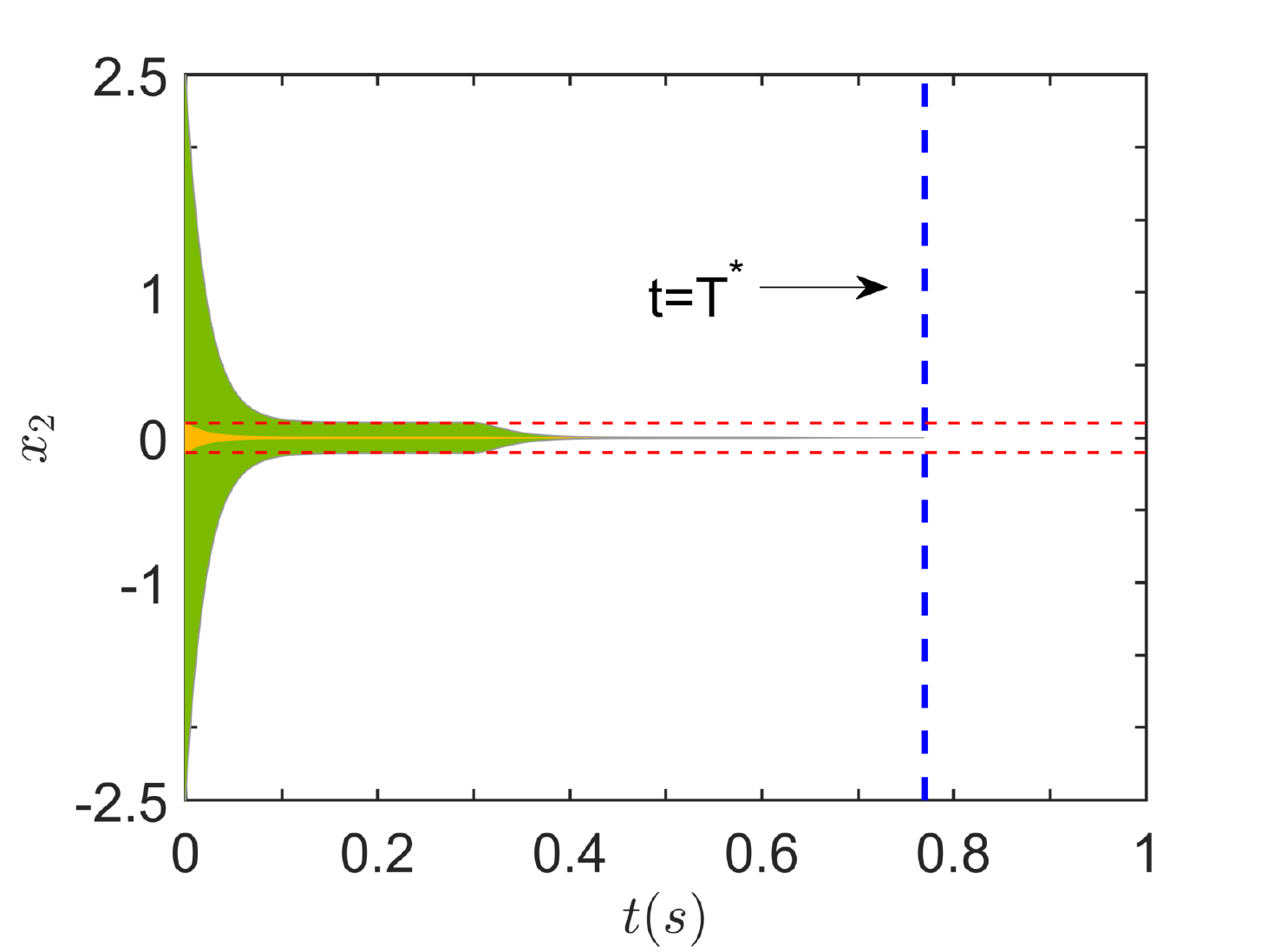}
    }
  
    \caption{In the low-pass filter system, the over-approximation of the reachable set of mode  $q_1$
    is shown in (a)\&(b), and the one of mode $q_2$ is shown in (c)\&(d). 
    All trajectories, marked with blue for mode $q_1$ (yellow for mode $q_2$), starting from the states contained in the first ball $\mathfrak{B}(\dfrac{w_{max}}{ \delta})$, are always enclosed in the second ball $\mathfrak{B}(\dfrac{w_{max}}{ \eta})$ denoted by  two red dashed lines.}
 \end{figure*}
 
 \begin{figure}[ht!]
     \centering
     \subfigure[$e_1$]{
        \label{fig:low1back}
    \includegraphics[width=4cm, height=3cm]{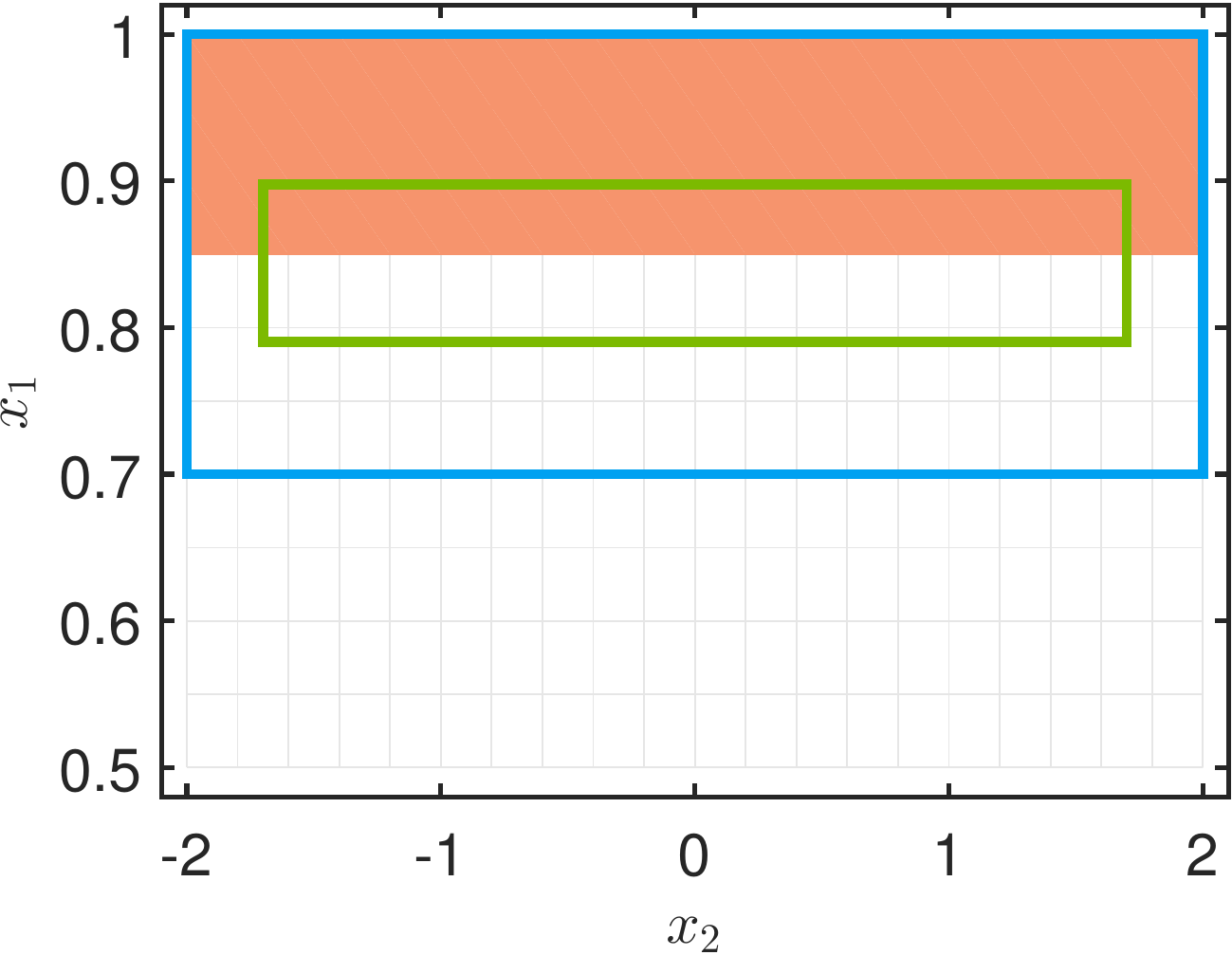}
    }
     \hspace{-3mm}
    \subfigure[$e_2$]{
        \label{fig:low2back}
    \includegraphics[width=4cm, height=3cm]{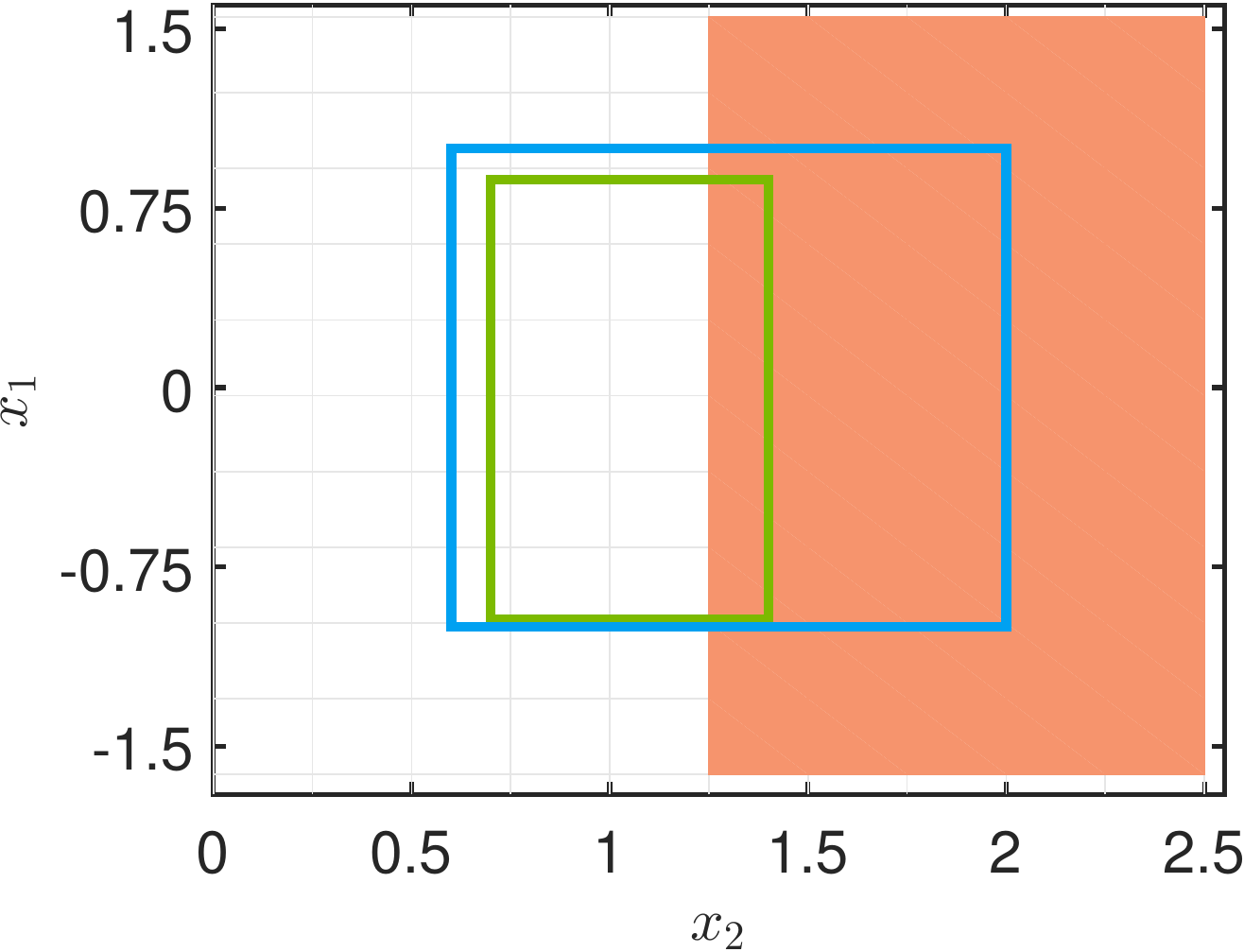}
    }
     \caption{The synthesized switching controller on the edge $e_1$ and $e_2$ of the low-pass filter system.  $\widetilde{G}$ is indicated by the blue box, and  $G^*$ is indicated by the red region. The green box stands for the forward reachable set in $0.01$s from  $G^*(e_1)$ ($0.02$s from $G^*(e_2)$). }
     \label{fig:low-pass_back}
 \end{figure}
 
\subsection{Low-pass Filter System}\label{case:1}

We first consider a low-pass filter system with delays, adapted from  CORA \cite{Althoff2016a}. 
It includes two first order low-pass filters $q_1$ and $q_2$,  represented by 
\begin{equation*}
q_1: \left\{
    \begin{array}{lr} 
        \left\{
        \begin{array}{lr}
   \dot{x}_1(t)= -14.58 x_1(t) + 2 x_1(t-0.1) + 0.5\sin(t)  \\
   \dot{x}_2(t) = -20.05 x_2(t) + 2 x_2(t-0.1) + 0.5\sin(t) 
    \end{array}
    \right. & \\
    \Xi(q_1)= [-1, 1]\times [-2, 2] & \\
    I(q_1)=\mathbb{R}^2, & 
    \end{array}
\right.  
\end{equation*}
\begin{equation*}
    q_2: \left\{
        \begin{array}{lr} 
            \left\{
            \begin{array}{lr}
        \dot{x}_1(t)= -32.66 x_1(t) + 8 x_1(t-0.1) + 0.5\sin(t)  \\
        \dot{x}_2(t)= -47.25 x_2(t) + 8 x_2(t-0.1) + 0.5\sin(t) 
            \end{array}
        \right. & \\
        
        \Xi(q_2)= [-2.25, 2.5]\times [-2.5, 2.5] & \\
        I(q_2)=\mathbb{R}^2. & 
        \end{array}
    \right.
\end{equation*}
There are two discrete transitions $e_1=(q_1, q_2)$ and $e_2=(q_2, q_1)$  between $q_1$ and $q_2$, and the  corresponding 
guard conditions  are $G(e_1)=\{(x_1, x_2)\in \mathbb{R}^2 \mid x_1 \geq 0.7 \}$,  $ G(e_2)=\{(x_1, x_2)\in \mathbb{R}^2 \mid x_2 \geq 0.6 \}$. Reset functions are identity mappings.
Moreover, both discrete transitions are taken with delays $D(e_1)=0.02$ and $D(e_2)=0.02$, respectively.
The safety requirement is  $\mathcal{S}=\{ (x_1, x_2)\in \mathbb{R}^2 \mid -2.7 \leq x_1 \leq 2.7 \wedge -2.6 \leq x_2 \leq 2.6\}$.

\begin{figure*}
    \centering
    \subfigure[]{
        \label{fig:precator1}
    \includegraphics[width=4.5cm, height=3cm]{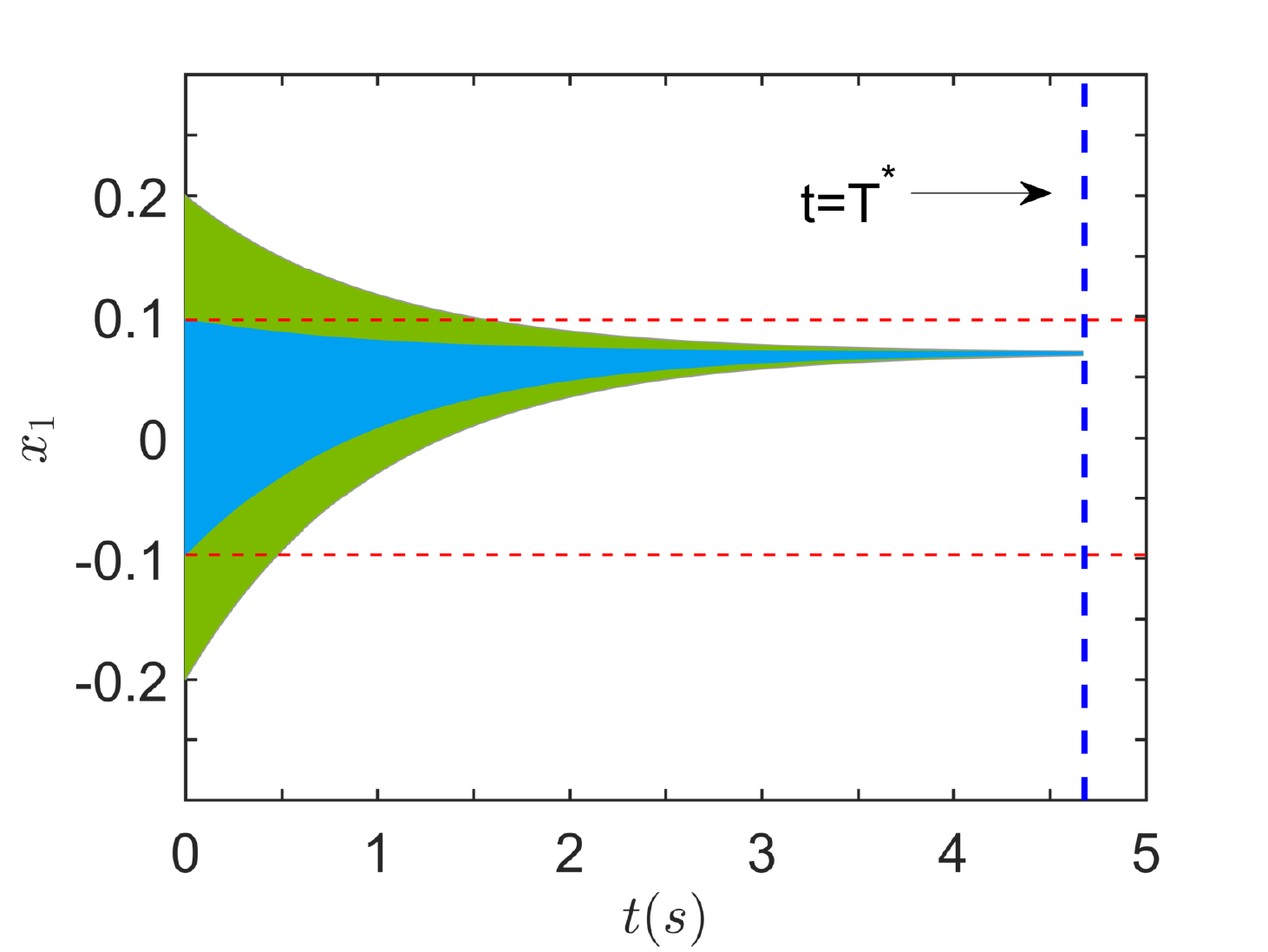}
    }
    \hspace{-7mm}
    \subfigure[]{
        \label{fig:precator2}
    \includegraphics[width=4.5cm,height=3cm]{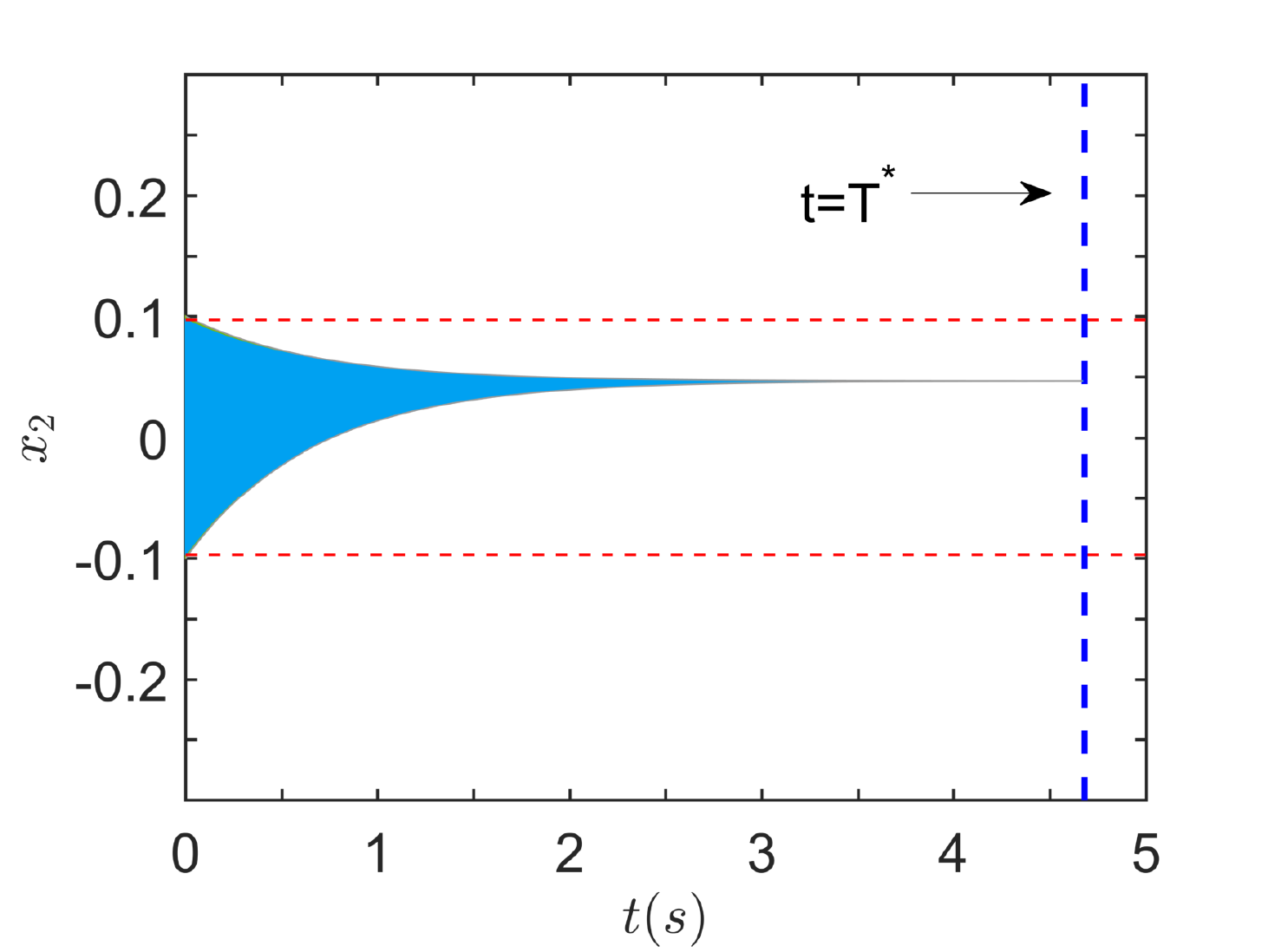}
    }
    \hspace{-7mm}
    \subfigure[]{
        \label{fig:precator3}
    \includegraphics[width=4.5cm, height=3cm]{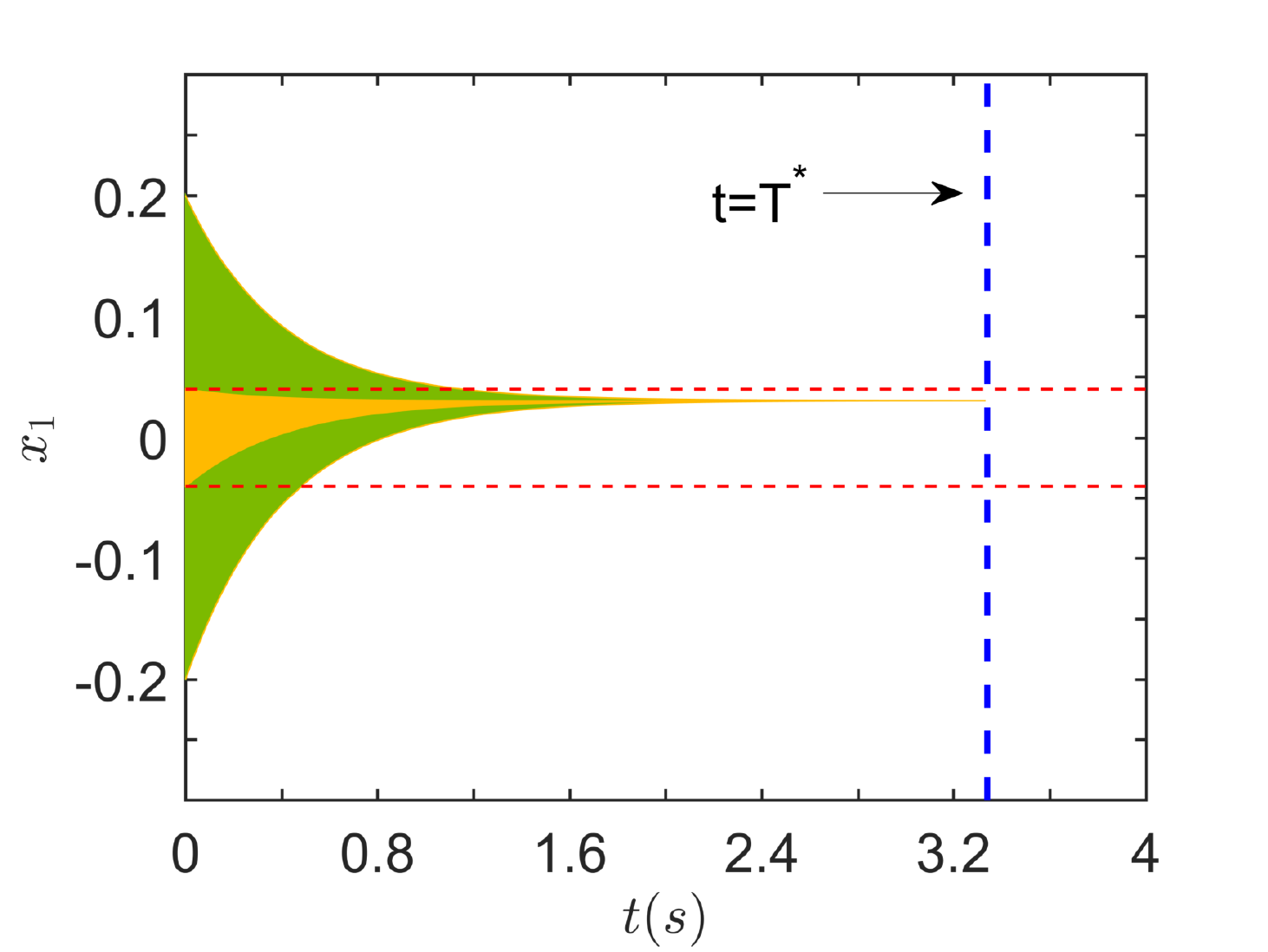}
    }
    \hspace{-7mm}
    \subfigure[]{
        \label{fig:precator4}
    \includegraphics[width=4.5cm,height=3cm]{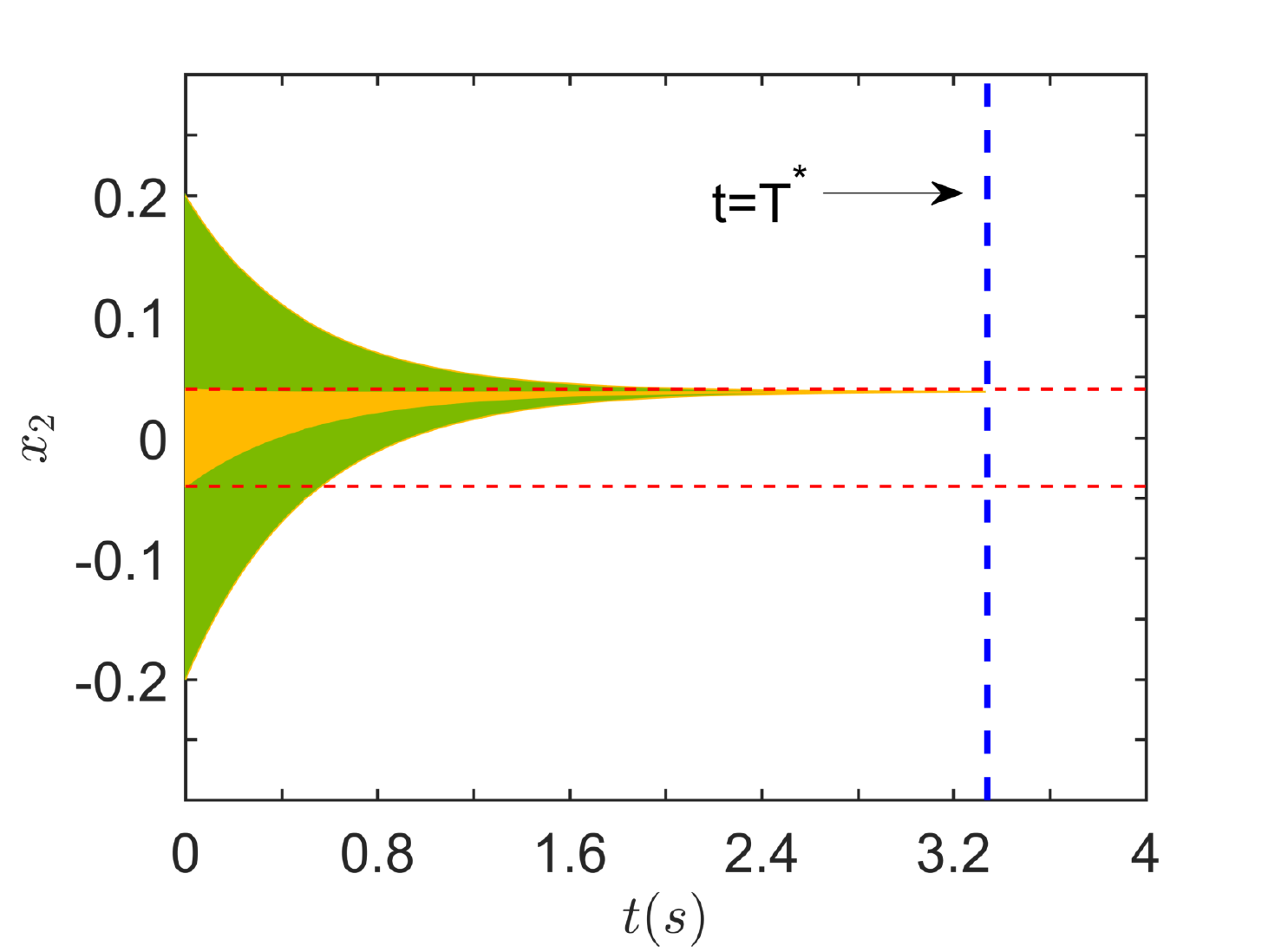} \vspace*{-2mm} 
    }
    \caption{ 
    In the predator-prey populations system, the over-approximation of the reachable set of mode  $q_1$
    is shown in (a)\&(b), and the one of mode $q_2$ is shown in  (c)\&(d). 
    All trajectories, marked with blue for mode $q_1$ (yellow for mode $q_2$), starting from the states contained in the first ball $\mathfrak{B}( \dfrac{\mathcal{G}}{ \delta})$,  are always enclosed in the second ball $\mathfrak{B}(\dfrac{\mathcal{G}}{ \eta})$ denoted by  two red dashed lines.}  
 \end{figure*}
 \begin{figure}
     \centering
     \subfigure[$e_1$]{
        \label{fig:precator1back}
    \includegraphics[width=4cm, height=3cm]{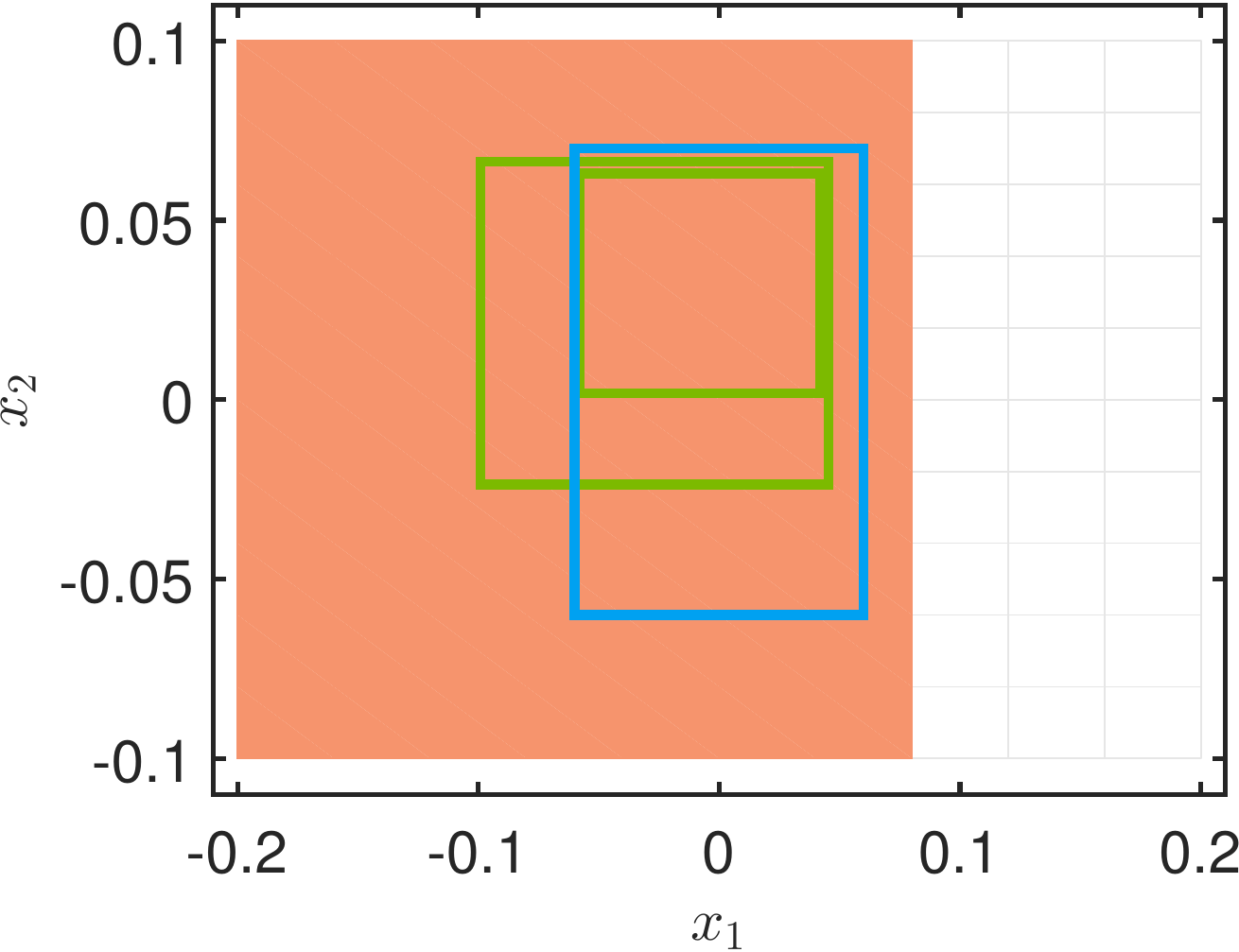}
    }
     \hspace{-3mm}
    \subfigure[$e_2$]{
        \label{fig:precator2back}
    \includegraphics[width=4cm, height=3cm]{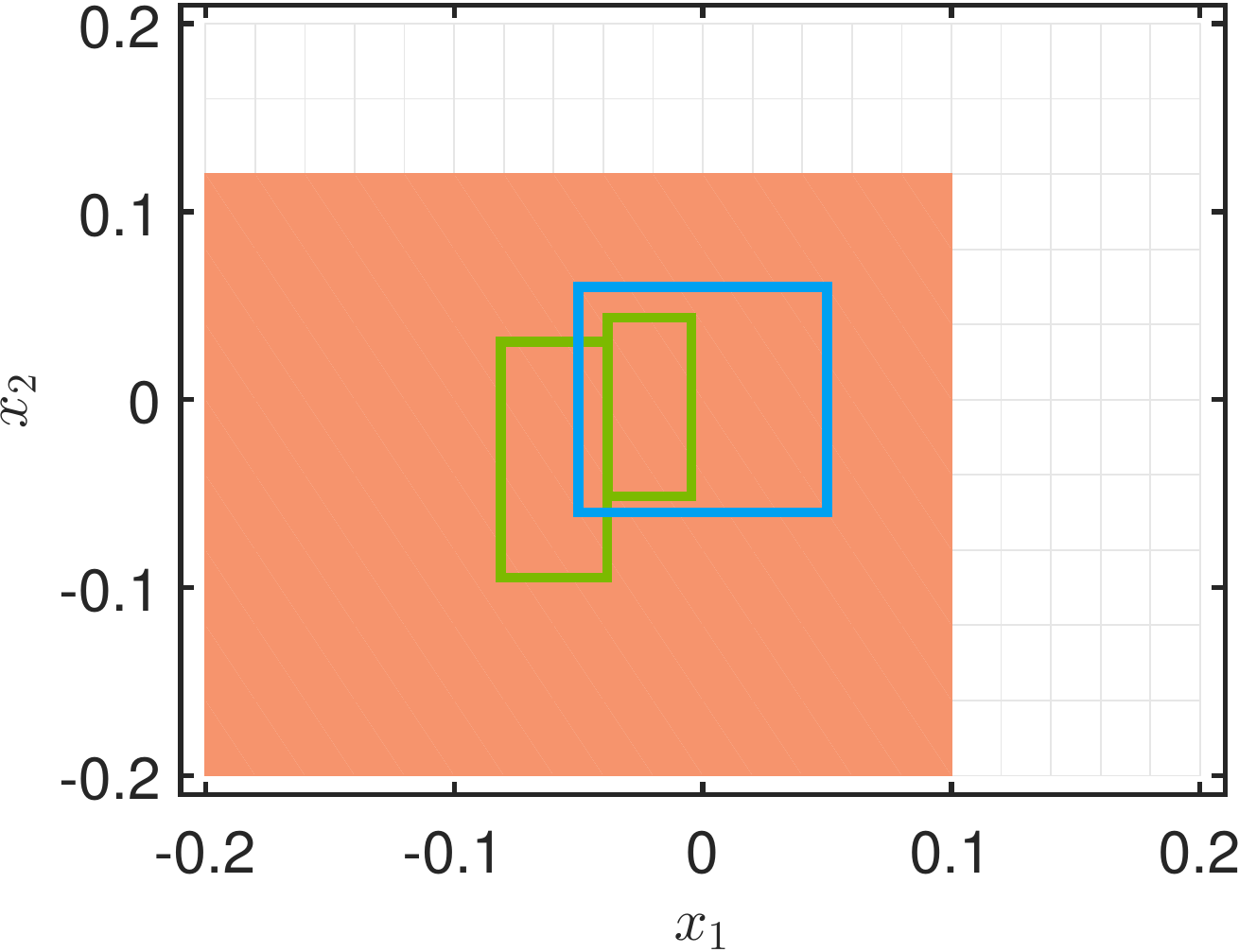} \vspace*{-2mm} 
    }
     \caption{The synthesized switching controllers on the edge $e_1$ and $e_2$ of the predator-prey population system.  $\widetilde{G}$ is indicated by the blue box, and  $G^*$  by the red region. The green boxes stand for the over-approximation of the reachable sets in $0.5$s and $0.8$s   from  $G^*(e_1)$ (in $0.3$s and $0.5$s from $G^*(e_2)$), respectively.}
     \label{fig:predator_back}
 \end{figure}
For mode $q_1$,
$M_{q_1}=\begin{bmatrix}
    -12.58 & 0 \\
    0 & -18.05
\end{bmatrix}$ is obviously a Metzler  matrix satisfying the two properties listed in Proposition~\ref{pro:Metzler}. By Theorem \ref{thm:reduceT}, 
the differential invariant synthesis problem is reduced to a $T^*_{q_1}$-differential invariant synthesis problem, 
where $T^*_{q_1}=0.5782s$ is computed with the parameters listed in Table \ref{tab:lowpass}.
Similarly, for mode $q_2$, 
$M_{q_2}=\begin{bmatrix}
    -4.66 & 0 \\
    0 & -39.25
\end{bmatrix}$ is also a Metzler  matrix satisfying the two properties listed in Proposition~\ref{pro:Metzler}.
 $T^*_{q_2}=0.7605s$ is computed with the  parameters listed in Table \ref{tab:lowpass}.
 The computed over-approximation of the reachable set within $T^*_{q_1}$ for mode $q_1$ using our approach 
 is given in Fig. \ref{fig:filter_a} and \ref{fig:filter_b}. 
The over-approximation of the reachable set in $T^*_{q_2}$ for mode $q_2$  is shown 
in Figure \ref{fig:filter_c} and \ref{fig:filter_d} with our approach.  Clearly,  the delay dynamical system in this mode  satisfies the ball convergence property.  
The guard conditions without discrete delays are 
$\widetilde{G}(e_1)=\{(x_1, x_2)\in \mathbb{R}^2 \mid  0.7\leq x_1\leq 1 \wedge -2\leq x_2 \leq 2\}$ and $\widetilde{G}(e_2)=\{(x_1, x_2)\in \mathbb{R}^2 \mid  -1 \leq x_1 \leq 1 \wedge 0.6 \leq x_2 \leq 2\}$. Finally,  applying Algorithm \ref{alg:backreachset}, 
the strengthened guard conditions $G^*(e_1)$ and  $G^*(e_2)$, that can guarantee the safety, 
are computed as showed in Fig. \ref{fig:low-pass_back}.

\subsection{Predator-prey Populations} \label{case:2}

We consider a  nonlinear predator-prey population dynamics under seasonal succession: a hybrid Lotka–Volterra competition model with delays adapted from \cite{Yanqing2017}. Two modes for two seasons are modelled as follows:
\begin{equation*}
q_1:\left\{
    \begin{array}{lr} 
        \left \{
        \begin{array}{lr}
     \dot{x}_1(t)  = - x_1(t)(1-\frac{x_1(t)}{100}) + 0.2 d_1 +w_{11}(t) \\
     \dot{x}_2(t)  =-1.5 x_2(t)(1-\frac{x_2(t)}{100} )+ 0.1 d_2 +w_{12}(t)
    \end{array}
    \right. & \\
    \Xi(q_1)= [-0.2, 0.2]\times [-0.1, 0.1] & \\
    I(q_1)=\mathbb{R}^2. & 
    \end{array}
\right.\\ 
\end{equation*}
\begin{equation*}
q_2: \left\{
        \begin{array}{lr} 
             \left\{
            \begin{array}{lr}
       \dot{x}_1(t)  = -2.5 x_1(t) + 0.2 x_1(t-0.01)(1+ x_2(t)) + w_{21}(t) \\
       \dot{x}_2(t)  = -2 x_2(t) + 0.15 x_2(t-0.01)(1+ x_2(t)) + w_{22}(t)
            \end{array}
        \right. & \\
        \Xi(q_2)= [-0.2, 0.2]\times [-0.2, 0.2] & \\
        I(q_2)=\mathbb{R}^2.& 
        \end{array}
    \right.
\end{equation*}
where $q_1$ and $q_2$ represent two seasons, $d_1=x_1(t-0.1)( 1+x_1(t))$, $d_2=x_2(t-0.1)( 1+x_2(t))$, 
$x_1$ is the number of prey  (for example, rabbits),
$x_2$ is the number of some predator  (for example, foxes),
$w_{ij}(t)= 0.07 \cos{2t}$ ($i, j = 1,2$) denote the perturbations. 
The real coefficients describe the interaction of the two species, the intrinsic growth rate and the environment capacity of the population in season $i$, respectively. 
There are two discrete transitions $e_1=(q_1, q_2)$ and $e_2=(q_2, q_1)$  between mode $q_1$ and mode $q_2$, and their corresponding 
guard conditions initially  are $G(e_1)=\{(x_1, x_2)\in \mathbb{R}^2 \mid -0.06\leq x_1\leq 0.06 \wedge -0.06 \leq x_2 \leq 0.07 \}$,  $ G(e_2)=\{(x_1, x_2)\in \mathbb{R}^2 \mid -0.05\leq x_1\leq 0.05 \wedge -0.06 \leq x_2 \leq 0.06 \}$.
Reset functions are identity mappings.
Moreover,  both discrete transitions are taken with delays $D(e_1)=1$ and $D(e_2)=0.55$, respectively.
The safety requirement is $\mathcal{S}= \{ (x_1, x_2)\in \mathbb{R}^2 \mid -0.20 \leq x_1 \leq 0.21 \wedge -0.21 \leq x_2 \leq 0.22 \}$.

By linearizing mode $q_1$, we have:
\begin{equation*}
\label{eq:pre_m1}
        \left\{
        \begin{array}{lr}
   \dot{x}_1(t) = - x_1(t)+ 0.2 x_1(x-0.1) +w_{11}(t) \\
   \dot{x}_2(t) = -1.5 x_2(t) +0.1 x_2(x-0.1)  +w_{12}(t)
    \end{array}
    \right..
\end{equation*}
Clearly, $M_{q_1}= \begin{bmatrix}
    -0.8 & 0 \\
    0 & -1.4
\end{bmatrix}$ is  a Metzler  matrix satisfying the two properties listed in Proposition~\ref{pro:Metzler}.
By Theorems \ref{thm:expenential_nonlinear} and \ref{thm:reduceT_nonlinear}, the differential invariant synthesis problem  for 
mode $q_1$ is reduce to a $T^*_{q_1}$-differential invariant synthesis problem, where  
$T^*_{q_1}= 4.6825$s is computed using our approach with the parameters listed in Table \ref{tab:prepre}.

Here it is noteworthy that $\iota = 0.2$, covering the entire initial set. 
Similarly, for mode $q_2$, the linearization of its dynamics is :
\begin{equation*}
               \left\{
            \begin{array}{lr}
       \dot{x}_1(t)= -2.5 x_1(t) +0.2 x_1(t-0.01)+ w_{21}(t) \\
       \dot{x}_2(t)= -2 x_2(t) + 0.15 x_2(t-0.01)+ w_{22}(t)
            \end{array}
        \right..
\end{equation*}
Clearly, $M_{q_2}= \begin{bmatrix}
    -2.3 & 0 \\
    0 & -1.85
\end{bmatrix}$ is also a Metzler  matrix satisfying the two properties listed in Proposition~\ref{pro:Metzler}. With the parameters listed in Table \ref{tab:prepre}, a bounded time  $T^*_{q_2} = 3.3326$s is computed.
The computed over-approximation of the reachable set within $t\geq T^*_{q_1}$ for mode $q_1$ is 
showed in Fig. \ref{fig:precator1}\&\ref{fig:precator2}. 
And the computed over-approximation of the reachable set within $t\geq T^*_{q_2}$ for mode $q_2$ are shown in Fig. \ref{fig:precator3}\&\ref{fig:precator4} using our approach. 
The  guard conditions without discrete delays are computed  as  
$\widetilde{G}(e_1)=\{(x_1, x_2)\in \mathbb{R}^2 \mid  -0.06\leq x_1\leq 0.06 \wedge -0.06 \leq x_2 \leq 0.07\}$ and $\widetilde{G}(e_2)=\{(x_1, x_2)\in \mathbb{R}^2 \mid  -0.05\leq x_1\leq 0.05 \wedge -0.06 \leq x_2 \leq 0.06 \}$. Finally,  applying Algorithm \ref{alg:backreachset}, 
the strengthened guard conditions  $G^*(e_1)$ and  $G^*(e_2)$, which can guarantee 
the safety requirement, are obtained as shown in Fig. \ref{fig:predator_back}.
\oomit{
\begin{table}[t]
    \centering
    \begin{tabular}{|c|c|c|c|c|c |c|}
    \hline
      Mode   & $\epsilon$ &$\zeta$ & $\beta$ & $\eta$ & $\gamma$ & $\delta$ \\ 
     \hline
      $q_1$  & 0.001  &  $\begin{bmatrix}
            1  \\ 
            1  
        \end{bmatrix}$     & 1    & 12.58     &5.16416  & 12.58 \\
             \hline
              $q_2$ &  0.001  &  $\begin{bmatrix}
            1  \\ 
            1  
        \end{bmatrix}$   &1   &  24.66   &4.22703   &24.66   \\
             \hline
        \end{tabular}
\captionsetup{justification=centering}
  \caption{The value of parameters  for the low-pass filter system}
    \label{tab:lowpass}
\end{table}
}

\section{Conclusion}\label{sec:conclu}
We introduced the notion of delay hybrid automata (dHA) in order to 
 model continuous 
delays and discrete delays in cyber-physical systems uniformly. Based on dHA, 
we proposed an approach on how to automatically synthesize a switching controller 
for a delay hybrid system with perturbations against 
a given safety requirement. To the end, we presented a new approach for over-approximating a nonlinear DDE with 
perturbation using ball-convergence analysis based on Metzler matrix. 
Two case studies were provided to indicate the effectiveness and efficiency of the proposed approach. 

For future work, it deserves to investigate how to synthesize a switching controller 
for a dHA against much richer properties defined e.g. by 
signal temporal logic \cite{STL} or metric temporal logic \cite{MTL}. In addition, it is interesting to consider 
 our method to deal with more general forms of DDEs. 
 Besides, it is a challenge how to guarantee the completeness of our approach, which essentially corresponds to a long-standing problem  on how to compute reachable sets of hybrid systems in the infinite time horizon.

\section*{Acknowledgements}
We thank Prof. Martin Fr\"{a}nzle, Dr. Mingshuai Chen and Mr. Shenghua Feng for fruitful discussions on this topic, and also thank the anonymous referees for their constructive comments and criticisms that improve this paper very much. 

The first, third and sixth authors are partly funded by NSFC-61625206 and NSFC-61732001, 
the second author is partly funded by NSFC-61902284, the fourth author is partly funded by NSFC-61732001, and the fifth author is partly funded by NSFC-61872341, NSFC-61836005 and the CAS Pioneer Hundred Talents Program.

\oomit{
\begin{table}
    \centering
    \begin{tabular}{|c|c|c|c|c|c |c|c|c|c|}
    \hline
      Mode   & $\epsilon$ &$\zeta$ & $\beta$ & $\eta$ & $\gamma$ & $\delta$   & $g_{max}$  & $\mathcal{G} $ & $\iota$   \\ 
     \hline
      $q_1$  & $10^{-4}$ &    $\begin{bmatrix}
            1  \\ 
            1  
        \end{bmatrix}$     & 1   & 0.8    &0.625  & 0.8 & 0.008    & 0.078   & 0.2  \\
             \hline
              $q_2$ &  $10^{-4}$   & $\begin{bmatrix}
            1  \\ 
          1  
        \end{bmatrix}$   &1   &  1.85  &0.88  & 1.85 & 0.005   & 0.075  &0.2 \\
             \hline
        \end{tabular}
 \captionsetup{justification=centering}
  \caption{The value of parameters\\ for the predator-prey populations system}
    \label{tab:prepre}
\end{table}}

\bibliographystyle{ACM-Reference-Format}
\bibliography{reportbib}

\end{document}